\documentclass[runningheads]{llncs}
\pdfoutput=1
\usepackage[T1]{fontenc}
\usepackage{graphicx}

\usepackage{subcaption}
\usepackage{enumitem}
\usepackage{cleveref}
\usepackage{listings}
\usepackage{multicol}
\usepackage{tikz}
\usepackage{color,soul}
\usepackage{wrapfig}
\usepackage{amssymb}

\setlist{nolistsep}

\newcommand{\LPList}{Link-and-Persist}

\newcommand{\LFList}{Link-Free} 
\newcommand{\SOFTList}{SOFT} 
\newcommand{\physdelList}{Physical-Delete}

\crefname{section}{\S}{\S\S}

\newcommand{\myparagraph}[1]{\noindent \textbf{#1.}}

\newcommand{\func}[1]{\textit{#1}}

\definecolor{codegreen}{rgb}{0,0.6,0}
\definecolor{codegray}{rgb}{0.5,0.5,0.5}
\definecolor{codepurple}{rgb}{0.58,0,0.82}
\definecolor{backcolour}{rgb}{0.95,0.95,0.92}
\definecolor{anti-flashwhite}{rgb}{0.95, 0.95, 0.96}

\lstdefinestyle{mystyle}{
  numberbychapter=false,
  mathescape=true,
  backgroundcolor=\color{anti-flashwhite},
  commentstyle=\color{codegray},
  keywordstyle=\color{magenta},
  numberstyle=\tiny\color{codegray},
  stringstyle=\color{codepurple},
  basicstyle=\ttfamily\footnotesize,
  breakatwhitespace=false,         
  breaklines=true,
  numberblanklines=false,
  captionpos=b,                    
  keepspaces=true,                 
  numbers=left,                    
  numbersep=5pt,                  
  showspaces=false,                
  showstringspaces=false,
  escapeinside={//}{\^^M},
  showtabs=false,                  
  tabsize=1,
  belowskip=-7mm, 
  keywordstyle=\color{codegreen},
  otherkeywords={for,each,true,false, class, struct, def}, morekeywords={type,subtype,break,continue,if,else,end,loop,while,do,done,exit, when,then,return,read,and,or,not,for,each,boolean,procedure,invoke,iteration,until,goto,wait}
}
\usepackage{fullpage}

\newtheorem{invariant}[theorem]{Invariant}
\crefname{invariant}{Invariant}{Invariants}

\definecolor{highlight}{rgb}{1.0, 0.0, 0.0} 
             
\begin{document}

\title{The Fence Complexity of Persistent Sets}

\author{Gaetano Coccimiglio$^1$ \and
Trevor Brown$^1$ \and
Srivatsan Ravi$^2$ \\
$^1$ University of Waterloo
$^2$ University of Southern California
}

\institute{}
\def\fullpaper{1}

\maketitle
\begin{abstract}
We study the psync complexity of concurrent sets in the non-volatile shared memory model. 
Flush instructions are used in non-volatile memory to force shared state to be written back to non-volatile memory and must typically be accompanied by the use of expensive fence instructions to enforce ordering among such flushes.
Collectively we refer to a flush and a fence as a psync.
The safety property of strict linearizability forces crashed operations to take effect before the crash or not take effect at all;
the weaker property of durable linearizability enforces this requirement only for operations that have completed prior to the crash event.
We consider lock-free implementations of list-based sets and prove two lower bounds. 
We prove that for any durable linearizable lock-free set there must exist an execution where some process must perform at least one redundant psync as part of an update operation. 
We introduce an extension to strict linearizability specialized for persistent sets that we call strict limited effect (SLE) linearizability.
SLE linearizability explicitly ensures that operations do not take effect after a crash which better reflects the original intentions of strict linearizability.
We show that it is impossible to implement SLE linearizable lock-free sets in which read-only (or search) operations do not flush or fence.
We undertake an empirical study of persistent sets that examines various algorithmic design techniques and the impact of flush instructions in practice. 
We present concurrent set algorithms that provide matching upper bounds and rigorously evaluate them against existing persistent sets to expose the impact of algorithmic design and safety properties on psync complexity in practice as well as the cost of recovering the data structure following a system crash.

%\keywords{Strict linearizability \and Durable linearizability \and Lower bounds \and Persistent sets \and Non-volatile memory.}
\end{abstract}
%
%
%
%\tableofcontents

\section{Introduction}

Byte-addressable Non-Volatile Memory (NVM) is now commercially available, thus accelerating the need for efficient \emph{persistent} concurrent data structure algorithms. 
We consider a model in which systems can experience full system crashes.
When a crash occurs the contents of volatile memory are lost but the contents of NVM remain persistent.
Following a crash a recovery procedure is used to bring the data structure back to a consistent state using the contents of NVM.
In order to force shared state to be written back to NVM the programmer is sometimes required to explicitly \emph{flush} shared objects to NVM by using \textit{explicit flush} and \textit{persistence fence} primitives, the combination of which is referred to as a \textit{psync} \cite{zuriel2019efficient}. 
While concurrent sets have been extensively studied for volatile shared memory~\cite{HS08-book}, they are still relatively nascent in non-volatile shared memory. 
This paper presents a detailed study of the psync complexity of \emph{concurrent sets} in theory and practice.

\myparagraph{Algorithmic design choices for persistent sets}
The recent trend is to persist less data structure state to minimize the cost of writing to NVM.
For example, the \LFList\ and \SOFTList~\cite{zuriel2019efficient} persistent list-based sets do not persist any \textit{pointers} in the data structure.
Instead they persist the \textit{keys} along with some other metadata used after a crash to determine if the key is in the data structure.
This requires at most a single psync for update operations;
however, not persisting the structure results in a more complicated recovery procedure.

A manuscript by Israelevitz and nine other authors presented a seminal in depth study of the performance characteristics of real NVM hardware~\cite{izraelevitz2019basic}.
Their results may have played a role in motivating the trend to persist as little as possible and reduce the number of fences.
In particular, they found (Figure~8 of~\cite{izraelevitz2019basic}) that the latency to write 256 bytes and then perform a psync is at least 3.5x the latency to write 256 bytes and perform a flush but no persistence fence.
Moreover, they found that NVM's write bandwidth could be a severe bottleneck, as a write-only benchmark (Figure~9 of~\cite{izraelevitz2019basic}) showed that NVM write bandwidth \textit{scaled negatively} as the number of threads increased past four, and was approximately \textit{9x lower} than volatile write bandwidth with 24 threads.
A similar study of real NVM hardware was presented by Peng et al. \cite{peng2019system}.

While these results are compelling, it is unclear whether these latencies and bandwidth limitations are a problem for concurrent sets in practice.
As it turns out, the push for \emph{persistence-free} operations and synchronization mechanisms that minimize the amount of data persisted, and/or the number of \emph{psyncs}, has many consequences, and the balance may favour incurring increased psyncs in some cases.

\myparagraph{Contributions}
Concurrent data structures in volatile shared memory typically satisfy the \emph{linearizability} safety property, NVM data structures must consider the state of the persistent object following a full system crash. 
The safety property of \emph{durable-linearizability} satisfies linearizability and following a crash, requires that the object state reflect a consistent operation subhistory that includes operations that had a response before the crash \cite{izraelevitz2016linearizability}. 
(i)
\ifx\fullpaper\undefined
We prove that for any durable-linearizable lock-free set there must exist an execution in which some process must perform at least one \emph{redundant} psync as part of an update operation (\cref{para:redundant}).
\else
We prove that for any durable-linearizable lock-free set there must exist an execution in which some process must perform at least one \emph{redundant} psync as part of an update operation (\cref{section:lb-redundant}).
\fi
Informally, a redundant psync is one that does not change the contents of NVM.
Our result is orthogonal to the lower bound of Cohen et al. who showed that the minimum number of psyncs per update for a durable-linearizable lock-free object is one \cite{cohen2018inherent}.
However, Cohen et al. did not consider redundant psyncs. 
We show that redundant psyncs cannot be completely avoided in all concurrent executions: there exists an execution where $n$ processes are concurrently performing update operations and $n-1$ processes perform a redundant psync.
(ii) 
Our first result also applies to \textit{SLE linearizability}, which we define to serve as a natural extension of the safety property of \textit{strict linearizability} specifically for persistent sets.
Originally defined by Aguilera and Frølund \cite{AF03}, strict linearizability forces crashed operations to be linearized before the crash or not at all.
Strict linearizability was not originally defined for models in which the system can recover following a crash.
% The adaptation of strict linearizability to models which allow for recovery does not succeed to capture the intuition of Aguilera and Frølund. 
To better capture the intentions of strict linearizability in the context of persistent concurrent sets, we introduce SLE linearizability to realize the intuition of Aguilera and Frølund for persistent concurrent sets.
SLE linearizability is defined to explicitly enforce \textit{limited effect} for persistent sets.

(iii) 
\ifx\fullpaper\undefined
We prove that it is impossible to implement SLE linearizable lock-free sets in which read-only operations neither flush nor execute persistence fences, but it is possible to implement strict linearizable and durable linearizable lock-free sets with persistence-free reads (\cref{para:sle}).
\else
We prove that it is impossible to implement SLE linearizable lock-free sets in which read-only operations neither flush nor execute persistence fences, but it is possible to implement strict linearizable and durable linearizable lock-free sets with persistence-free reads (\cref{section:lb-sle}).
\fi
(iv) 
We study the empirical costs of persistence fences in practice. 
To do this, we present matching upper bounds to our lower bound contributions (i) and (ii).
Specifically, we describe a new technique for implementing persistent concurrent sets with persistence-free read-only operations called the extended link-and-persist technique and we utilize this technique to implement several persistent sets (\cref{section:algos}).
(v) 
We evaluate our upper bound implementations against existing persistent sets in a systemic empirical study of persistent sets.
This study exposes the impact of algorithmic design and safety properties on persistence fence complexity in practice and the cost of recovering the data structure following a crash (\cref{section:eval}).

The relationship between performance, psync complexity, recovery complexity and the correctness condition is subtle, even for seemingly simple data types like sorted sets. 
In this paper, we delve into the details of algorithmic design choices in persistent data structures to begin to characterize their impact.

\ifx\fullpaper\undefined
\section{Lower Bounds}
\label{sec:lbs}
\myparagraph{Persistency Model and Safety Properties }
We assume a full system crash-recovery model (all processes crash together).
When a crash occurs all processes are returned to their initial states.
After a crash a recovery procedure is invoked, and only after that can new operations begin.

Modifications to base objects first take effect in the volatile shared memory.
Such modifications become persistent only once they are flushed to NVM. 
Base objects in volatile memory are flushed asynchronously by the processor (without the programmer's knowledge) to NVM arbitrarily.
We refer to this as a \textit{background flush}.
We consider \textit{background flushes} to be atomic.
The programmer can also \textit{explicitly} flush base objects to NVM by invoking \textit{flush} primitives, typically accompanied by \textit{persistence fence} primitives.
An \textit{explicit flush} is a primitive on a base object and is non-blocking, i.e., it may return \textit{before} the base object has been written to persistent memory.
An \textit{explicit flush} by process $p$ is guaranteed to have taken effect only after a subsequent \textit{persistence fence} by $p$.
An explicit flush combined with a persistence fence is referred to as a \textit{psync}.
We assume that psync events happen independently of RMW events and that psyncs do not change the configuration of volatile shared memory (other than updating the program counter).
Note that on Intel platforms a RMW implies a fence, however, a RMW does not imply a flush before that fence, and therefore does not imply a psync.

In this paper, we consider the \emph{set} type: an object of the set type stores a set of integer values, initially empty, and exports three operations: 
$\texttt{insert}(v)$, $\texttt{remove}(v)$, $\texttt{contains}(v)$ where $v \in \mathbb{Z}$.
A \emph{history} is a sequence of invocations and responses of operations on the set implementation. We say a history is well-formed if no process invokes a new operation before
the previous operation returns.
Histories $H$ and $H'$ are \emph{equivalent} if for every process
$p_i$, $H|i=H'|i$.

A history $H$ is durable linearizable, if it is well-formed and if $ops(H)$ is linearizable where $ops(H)$ is the subhistory of $H$ containing no crash events \cite{izraelevitz2016linearizability}.  

Aguilera and Frølund defined strict linearizability for a model in which individual processes can crash and did not allow for recovery \cite{AF03}.
Berryhill et al. adapted strict linearizability for a model that allows for recovery following a system crash \cite{golab15}.
A history $H$ is \emph{strict linearizable} with 
respect to an object type $\tau$ if there exists
a sequential history $S$ equivalent to a \emph{strict completion of $H$}, such that
(1) $\rightarrow_{H^c}\subseteq \rightarrow_S$ and
(2) \emph{$S$ is consistent with the sequential specification of $\tau$}.
A strict completion of $H$ is obtained from $H$ by inserting matching responses for a subset of pending operations after the operation’s invocation and before the next crash event (if any), and finally removing any remaining pending operations and crash events.

\myparagraph{Psync Complexity}
It is likely that an implementation of persistent object will have many similarities to a volatile object of the same abstract data type.
For this reason, when comparing implementations of persistent objects we are mostly interested in the overhead required to maintain a consistent state in persistent memory. 
Specifically, we consider psync complexity.

Programmers write data to persistent memory through the use of psyncs.
A psync is an expensive operation.
Cohen et al. \cite{cohen2018inherent} prove that update operations in a durable linearizable lock-free algorithm must perform at least one psync.
In some implementations of persistent objects, reads also must perform psyncs.
There is a clear focus in existing literature on minimizing the number of psyncs per data structure operation \cite{david2018log,zuriel2019efficient,friedman2020nvtraverse}. 
These factors suggest that psync complexity is a useful metric for comparing implementations of persistent objects.

\myparagraph{Lower Bounds for Persistent Sets} We now present the two main lower bounds in this paper, but the full proofs are only provided in the full version of the paper\cite{fullpersistentsetspaper} due to space constraints.

\myparagraph{\underline{Impossibility of persistence-free read-only searches}}
\label{para:sle}
The key goal of the original work of Aguilera and Frølund~\cite{AF03} was to enforce \textit{limited effect} by requiring operations to take effect before the crash or not at all. 
Limited effect requires that an operation takes effect within a limited amount of time after it is invoked.
The point at which an operation takes effect is typically referred to as its \textit{linearization point} and we say that the operation \textit{linearizes} at that point.
Rephrasing the intuition, when crashes can occur, limited effect requires that operations that were pending at the time of a crash linearize prior to the crash or not at all.

Strict linearizability is defined in terms of histories, which abstract away the real-time order of events.
As a result, strict linearizability does not allow one to argue anything about the ordering of linearization points of operations that were pending at the time of a crash relative to the crash event.
Thus, strict linearizability cannot and does not prevent operations from taking effect during the \emph{recovery procedure} or even after the recovery procedure (which can occur for example in implementations that utilize linearization helping~\cite{help15}).
Strict linearizability only requires that at the time of a crash, pending operations \textit{appear} to take effect prior to the crash.
Although we are not aware of a formal proof of this, we conjecture in the full system crash-recovery model, durable linearizable objects are strict linearizable for some suitable definition of the recovery procedure.
This is because we can always have the recovery procedure \textit{clean-up} the state of the object returning it to a state such that the resulting history of any possible extension will satisfy strict linearizability.
We note this conjecture as further motivation towards re-examining the way in which the definition of strict linearizability has been adapted for the full system crash-recovery model.

To this end, we define the concept of a \emph{key write} to capture the intentions of Aguilera and Frølund in the context of sets by defining \textit{Strict limited effect} (SLE) linearizability for sets as follows:
a history satisfies SLE linearizability iff the history satisfies strict linearizability and for all operations with a key write, if the operation is pending at the time of a crash, the key write of the operation must occur before the crash event.
In the strict completion of a history this is equivalent to requiring that the key write is always between the invocation and response of the operation. 
This is because the order of key writes relative to a crash event is fixed which means if the write occurs after the crash event then a strict completion of the history could insert a response for the operation only prior to the key write (at the crash) and this response cannot be reordered after the key write.
  
We show that it is impossible to implement a SLE linearizable lock-free set for which read-only searches do not perform any explicit flushes or persistence fences. 
\begin{theorem}
There exists no SLE linearizable lock-free set with \emph{persistence-free} read-only searches.
\label{theorem:persistence-free-search-impossiblity}
\end{theorem}

\myparagraph{\underline{Redundant psync lower bound for durable linearizable sets}}
\label{para:redundant}
After modifying a base object only a single psync is required to ensure that it is written to persistent memory. 
Performing multiple psyncs on the same base object is therefore unnecessary and wasteful.
We refer to these unnecessary psyncs as \textit{redundant} psyncs.
We show that for any durably linearizable lock-free set there must exist an execution in which $n$ concurrent processes are invoking $n$ concurrent update operations and $n$-1 processes each perform at least one redundant psync.
At first glance one may think that this result is implied by the lower bound of Cohen et al. \cite{cohen2018inherent}.
Cohen et al. show that for any lock-free durable linearizable object, there exists an execution wherein every update operation performs at least one persistence fence.
Cohen et al. make no claims regarding redundant psyncs. 
Our result demonstrates that durable linearizable lock-free objects cannot completely avoid redundant psyncs. 
\begin{theorem}
In an $n$-process system, for every durable linearizable lock-free set implementation $I$, there exists an execution of $I$ wherein $n$ processes are concurrently performing update operations and $n$-1 processes perform a redundant psync. 
\label{theorem:redundant-psync}
\end{theorem}
\else
\section{Computational Model}
\label{section:model}

We present preliminaries for the standard \emph{volatile} shared memory model and then explain how we extend the model to \emph{non-volatile} (or \emph{persistent}) shared memory.

\myparagraph{Processes and shared memory}
We consider an asynchronous shared memory system in which a set of $\mathbb{N}$ processes communicate by applying \emph{operations} on shared \emph{objects}.
Each process $p_i;i\in \mathbb{N}$ has an unique identifier and an initial state.
An object is an instance of an \emph{abstract data type} which specifies a set of operations that provide the only means to
manipulate the object.
This paper considers the \emph{set} type: an object of the set type stores a set of integer values, initially empty, and exports three operations: 
$\texttt{insert}(v)$, $\texttt{remove}(v)$, $\texttt{contains}(v)$ where $v \in \mathbb{Z}$.
The update operations, $\texttt{insert}(v)$ and $\texttt{remove}(v)$, return a Boolean response, $\texttt{true}$ iff $v$ is absent (for $\texttt{insert}(v)$) or present (for $\texttt{remove}(v)$) in the list.  
After $\texttt{insert}(v)$ is complete, $v$ is present in the list, and after $\texttt{remove}(v)$ is complete, $v$ is absent from the list.
The $\texttt{contains}(v)$ returns a Boolean $\texttt{true}$ iff $v$ is present in the list.
Throughout the paper we refer to $\texttt{contains}$ operations as searches or read-only operations.

An \emph{implementation} of an object type (sometimes we say object) provides a specific data-representation by applying \emph{primitives} on a set of shared \emph{base objects} each of which has an initial value.
We assume that the primitives applied on base objects are \emph{deterministic}.
A primitive is a generic \emph{read-modify-write} (\emph{RMW}) procedure applied to a base object.

\myparagraph{Executions and configurations}
An \emph{event} of a process $p_i$ in the volatile shared memory model (sometimes we say \emph{admissible step} of $p_i$)
is an invocation or response of an operation performed by $p_i$ or a 
RMW primitive applied by $p_i$ to a base object
along with its response. 
A \emph{configuration} specifies the value of each base object and 
the state of each process.
The \emph{initial configuration} is the configuration in which all 
base objects have their initial values and all processes are in their initial states.

An \emph{execution fragment} is a (finite or infinite) sequence of events.
An \emph{execution} of an implementation $I$ is an execution
fragment where, starting from the initial configuration, each event is issued according to $I$ and each response of a RMW event on the base object $b$ matches the state of $b$ resulting from all preceding events on $b$.
A \emph{history} $H$ of an execution $E$ is the subsequence of $E$ consisting of all
invocations and responses of operations.

An execution $E\cdot E'$, denoting the concatenation of $E$ and $E'$,
is an \emph{extension} of $E$ and we say that $E'$ \emph{extends} $E$.
For every process identifier $k$, $E|k$ denotes the (possibly empty) subsequence of the execution $E$ restricted to events of process $p_k$.
An operation $\pi$ \emph{precedes} another operation $\pi'$ in an execution $E$, denoted $\pi \rightarrow_{E} \pi'$, if the response of $\pi$ occurs before the invocation of $\pi'$ in $E$.
Two operations are \emph{concurrent} if neither precedes
the other. 
An execution is \emph{sequential} if it has no concurrent 
operations. 
An operation $\pi_k\in E$ is \emph{complete in $E$} if
it returns a matching response in $E$.
Otherwise we say that it is \emph{incomplete} or \emph{pending} in $E$.
We say that an execution $E$ is \emph{complete} if every invoked operation is complete in $E$.
Note that all of the terminology defined above applies analogously to histories.
An implementation $I$ is \emph{lock-free} if it guarantees that in every execution $E$ of $I$ some process will always make progress by completing its operation within a finite number of its own steps.

\myparagraph{Well-formed executions} We assume that executions are \emph{well-formed}:
no process invokes a new operation before
the previous operation returns.

\myparagraph{Linearizability}
Histories $H$ and $H'$ are \emph{equivalent} if for every process
$p_i$, $H|i=H'|i$.
A complete history $H$ is \emph{linearizable} with 
respect to an object type $\tau$ if there exists
a sequential history $S$ equivalent to $H$ such that
(1) $\rightarrow_{H}\subseteq \rightarrow_S$ (the happens before order is preserved) and
(2) \emph{$S$ is consistent with the sequential specification of type $\tau$}.
A history $H$ is linearizable if it can be
\emph{completed} (by adding matching responses to a subset of incomplete operations in $H$ and removing the rest) to a linearizable history~\cite{HW90}.

\myparagraph{Decided Operation}
Consider an execution $E$ of a durable linearizable set.
The response of a pending update operation $\pi$ is \textbf{decided} in $E$ if for every possible crash-free extension of $E$ the response of $\pi$ is the same value $V$.
We say its \textbf{decided response} is $V$.

\myparagraph{Successful Operation}
Consider an execution $E$ of a durable linearizable set.
We say that an update operation is successful (resp., unsuccessful) if: (1) it has completed and its response is \texttt{true} (resp., \texttt{false}), or (2) if its decided response is \texttt{true} (resp., \texttt{false}).

\myparagraph{Linearization Help-Freedom}
We say that $f$ is a linearization function over a set of histories $\mathcal{H}$, if for every $H \in \mathcal{H}$, $f(H)$ is a linearization of $H$.
An execution of $E$ is linearization help-free there exists a linearization function $f$ over $E$ such that for any two operations $\pi_1, \pi_2 \in E$ and a single step $\gamma$, it holds that if $\pi_1$ is decided before $\pi_2$ in $E \cdot \gamma$ and $\pi_2$ is not decided before $\pi_1$ in $E$ then $\gamma$ is a step of $\pi_1$ by the process that invoked $\pi_1$.
An implementation $I$ is linearization help-free if all executions of $I$ are linearization help-free \cite{ben2020separation}.

\myparagraph{Persistence model} 
We assume a full system crash-recovery model (all processes crash together).
When a crash occurs all processes are returned to their initial states.
After a crash a recovery procedure is invoked, and only after that can new operations begin.

Modifications to base objects first take effect in the volatile shared memory.
Such modifications become persistent only once they are flushed to NVM. 
Base objects in volatile memory are flushed asynchronously by the processor (without the programmer's knowledge) to NVM arbitrarily.
We refer to this as a \textit{background flush}.
We consider \textit{background flushes} to be atomic.
The programmer can also \textit{explicitly} flush base objects to NVM by invoking \textit{flush} primitives, typically accompanied by \textit{persistence fence} primitives.
An \textit{explicit flush} is a primitive on a base object and is non-blocking, i.e., it may return \textit{before} the base object has been written to persistent memory.
An \textit{explicit flush} by process $p$ is guaranteed to have taken effect only after a subsequent \textit{persistence fence} by $p$.
An explicit flush combined with a persistence fence is referred to as a \textit{psync}.
We assume that psync events happen independently of RMW events and that psyncs do not change the configuration of volatile shared memory (other than updating the program counter).
Note that on Intel platforms a RMW implies a fence, however, a RMW does not imply a flush before that fence, and therefore does not imply a psync.

\myparagraph{Durable linearizability}
A history $H$ is durable linearizable, if it is well-formed and if $ops(H)$ is linearizable where $ops(H)$ is the subhistory of $H$ containing no crash events \cite{izraelevitz2016linearizability}.  

\myparagraph{Strict linearizability}
Aguilera and Frølund defined strict linearizability for a model in which individual processes can crash and did not allow for recovery \cite{AF03}.
Berryhill et al. adapted strict linearizability for a model that allows for recovery following a system crash \cite{golab15}.
A history $H$ is \emph{strict linearizable} with 
respect to an object type $\tau$ if there exists
a sequential history $S$ equivalent to a \emph{strict completion of $H$}, such that
(1) $\rightarrow_{H^c}\subseteq \rightarrow_S$ and
(2) \emph{$S$ is consistent with the sequential specification of $\tau$}.
A strict completion of $H$ is obtained from $H$ by inserting matching responses for a subset of pending operations after the operation’s invocation and before the next crash event (if any), and finally removing any remaining pending operations and crash events.

\section{Background}
\label{section:background}
In this section we will provide some necessary background information regarding metrics used to compare persistent objects and present some existing persistent sets.

\subsection{Complexity Measures}
\label{section:complexity-measures}
It is likely that an implementation of persistent object will have many similarities to a volatile object of the same abstract data type.
For this reason, when comparing implementations of persistent objects we are mostly interested in the overhead required to maintain a consistent state in persistent memory. 
Specifically, we consider psync complexity and recovery complexity.

\myparagraph{Psync Complexity}
Programmers write data to persistent memory through the use of psyncs.
A psync is an expensive operation.
Cohen et al. \cite{cohen2018inherent} prove that update operations in a durable linearizable lock-free algorithm must perform at least one psync.
In some implementations of persistent objects, reads also must perform psyncs.
There is a clear focus in existing literature on minimizing the number of psyncs per data structure operation \cite{david2018log,zuriel2019efficient,friedman2020nvtraverse}. 
These factors suggest that psync complexity is a useful metric for comparing implementations of persistent objects.

\myparagraph{Recovery Complexity}
After a crash, a recovery procedure is invoked to return the objects in persistent memory back to a consistent state.
Prior work has utilized a sequential recovery procedure \cite{zuriel2019efficient,david2018log,friedman2021mirror,correia2020persistent}. 
A sequential recovery procedure is not required for correctness but it motivates the desire for efficient recovery procedures.
No new data structure operations can be invoked until the recovery procedure has completed. 
Ideally we would like to minimize this period of downtime represented by the execution of recovery procedure. 
We use the asymptotic time complexity of the recovery procedure as another metric for comparing durable linearizable algorithms.

\subsection{Related Work}
In this section we will briefly describe some existing implementations of persistent sets. 
We focus on hand-crafted implementations since they generally perform better in practice compared to transforms or universal constructions.

\myparagraph{\textbf{Link-and-Persist (L\&P)}} 
David et al. describe a technique for implementing durable linearizable link-based data structures called the \LPList~technique \cite{david2018log}.
Using the \LPList~technique, whenever a link in the data structure is updated, a single bit mark is applied to the link which denotes that it has not been written to persistent memory.
The mark is removed after the link is written to persistent memory.
We refer to this mark as the \textit{persistence bit}.
This technique was also presented by Wang et al. in the same year \cite{wang2018easy}.
Wei et al. presented a more general technique which does not steal bits from data structure links \cite{wei2021flit}.

\myparagraph{\textbf{Link-Free (LF)}}
The \LFList~algorithm of Zuriel et al. does not persist data structure links \cite{zuriel2019efficient}. 
Instead, the \LFList~algorithm persists metadata added to every node.  

\myparagraph{\SOFTList}
Zuriel et al. designed a different algorithm called \SOFTList~(Sets with an Optimal Flushing Technique) offering persistence-free searches. 
The \SOFTList~algorithm does not persist data structure links and instead persists metadata added to each node. 
The major difference between the \LFList~algorithm and \SOFTList~is that \SOFTList~uses two different representations for every key in the data structure where only one representation is is explicitly flushed to persistent memory. 

\myparagraph{Transforms}
Friedman et al. presented a transform for converting a class of data structures which they call traversal data structures to durable linearizable data structures \cite{friedman2020nvtraverse}..
NVTraverse does not perform flushes during traversal of the data structure.
In a separate paper Friedman et al. presented Mirror which is automatic transform that converts a linearizable lock-free data structure to a durable linearizable lock-free data \cite{friedman2021mirror}
Mirror maintains two copies of the data structure, only one of which is persisted. 
Reads are executed on the transient data structure which is stored in DRAM. 

\myparagraph{Universal Constructions}
The Order Now, Linearize Later (ONLL) universal construction from Cohen et al. \cite{cohen2018inherent} transforms a deterministic object and produces a lock-free durably-linearizable implementation of the object.
With ONLL search operations do not perform psyncs however, it relies on a shared global queue and it must traverse logs from every process to recovery the data structure.
CX-PUC from Correia et al. is the first bounded wait-free persistent universal construction \cite{correia2020persistent}.
CX-PUC suffers from the fact that it must persist multiple replicas of the data structure.

\myparagraph{Transactional Memory}
There have been several works that proposed using transactional-memory (TM) with NVM to achieve persistent data structures.
Typically these approaches utilize some form of logging and generally perform better than universal constructions \cite{correia2020persistent,kolli2016high,avni2016persistent,coburn2011nv}
\section{Redundant psync lower bound for durable linearizable sets}
\label{section:lb-redundant}

\Cref{fig:lb2-executions} shows an example of an execution in which all but one process performs a redundant psync.
This execution is described in more detail in the proof of \cref{theorem:redundant-psync}.

\begin{figure}[!t]
\begin{center}
    \includegraphics[width=0.6\linewidth]{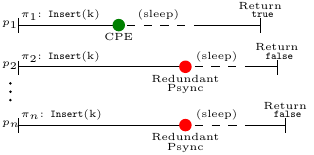}
\caption{An example of the execution described in \cref{theorem:redundant-psync}.
Operations $\pi_2-\pi_{n-1}$ rely on the durability of $\pi_1$ forcing.
Since none of these processes can know if the CPE of $\pi_1$ has occurred they are forced to perform a psync resulting each performing a redundant psync.} 
\label{fig:lb2-executions}
\end{center}
\end{figure}

\begin{definition}[Persistence Event]
Let $E$ be an execution of a durable linearizable set, we call any background flush, explicit flush or persistence fence a \textit{persistence event} in $E$. 
\end{definition}

\begin{definition}[Crash-Recovery Extension]
Consider an execution $E$ of a durable linearizable set.
Let $\bot$ denote a system crash event.
Let $E'$ be $E \cdot \bot \cdot E_R$ where $E_R$ is the sequential execution of the recovery procedure.
We refer to $E'$ as a \textbf{crash-recovery extension} of $E$.
\end{definition}

For a successful update operation $\pi$, we can identify a \textit{critical persistence event}~$e$.
Intuitively, if a crash event happens \textit{before} $e$, then the update $\pi$ will \textit{not be recovered}.
This means if we perform an identical update after recovery, it will \textit{succeed}.
On the other hand, if a crash happens \textit{after} $e$, then $\pi$ \textit{will be recovered}.
So, if we perform an identical update after recovery, it will \textit{fail}.

\begin{definition}[Critical Persistence Event (CPE)]
Consider an execution $E$ of a durable linearizable set and a pending operation $\pi$ in $E$ invoked by process $p$.
Let $r$ be the response of $\pi$ in a solo extension of $E$ wherein $p$ completes $\pi$ and no crash events occur.
Let $E'$ be the crash-recovery extension of $E$.
Consider a solo extension of $E'$ wherein $p$ invokes and completes a new operation $\pi'$ with the same arguments as $\pi$.
Let $r'$ be the response of $\pi'$.
For successful update operations, the persistence event $f$ in $E$ is the \textit{critical persistence event}, of $\pi$ if immediately before applying $f$, we have $r = r' = \texttt{true}$ and immediately after applying $f$ we have $r = \texttt{true}$ and $r' = \texttt{false}$.
\end{definition}

A CPE is defined for a successful update operation, and it represents the first point at which $\pi$ could, but is not guaranteed to, change the response of a different operation $\pi'$ where $\pi'$ is invoked after crashing and recovering. 
Crucially, $\pi$ and $\pi'$ are different operations, and $\pi'$ is the operation that completes in a solo extension after the crash. 
Note that immediately after recovering from a crash $\pi'$ has not yet started, so we cannot (and do not need to) argue that all extensions return the same response value.

Since the CPE is defined for successful update operations, if the CPE exists for an update $\pi$ in some execution $E$, then the CPE of $\pi$ exists at some point between when $\pi$ is decided and the response of $\pi$ (or a crash event). 
% The CPE of an update operation cannot be an explicit flush on a base object $b$ since an explicit flush requires a persistence fence to guarantee that that $b$ is written to persistent memory. 
% This means the CPE of an operation is either a background flush or a persistence fence. 
Recall that since explicit flushes are non-blocking, an explicit flush on a base object $b$ requires a persistence fence to guarantee that that $b$ is written to persistent memory.
It is possible that $b$ is written to persistent memory after an explicit flush on $b$ but prior any persistence fences.
In such a scenario the CPE of the operation can be the explicit flush on $b$.
% If $b$ is written to persistent memory after a persistence fence which followed an explicit flush on $b$ then the CPE of the operation is the persistence fence. 
If the CPE $\pi$ is a persistence fence by process $p$, then $p$ must have previously performed an explicit flush in $\pi$. 
We say that the base object $b$ is \textit{involved in the CPE} of $\pi$ if the CPE is a background flush or explicit flush on $b$ or if the CPE is a persistence fence where the corresponding explicit flush is on $b$.
% Unsuccessful update operations and read-only operations can be thought of as having a CPE.
% The CPE of an unsuccessful update operation or read-only operation corresponds to the CPE of some successful update operation.
% The corresponding CPE for an unsuccessful update operation or read-only operation can be prior to the invocation of the operation (this is easiest to see in the sequential case).

Our CPE definition is similar to the \textit{Persist Point} defined by Izraelevitz et al. \cite{izraelevitz2016linearizability} or the \textit{Durability Point} defined by Friedman et al \cite{QueueFriedman18}.
The \textit{Persist Point} of an operation is a point after the linearization point of the operation.
We have already mentioned two examples of implementations where the linearization point is after the point when the operation is persisted \cite{zuriel2019efficient,cohen2018inherent}.
The definition of a \textit{Durability Point} is not defined in terms of sets and lacks some necessary details which our definition of the CPE clarifies.

\begin{definition}[Destructive Write]
A write to the base object $b$ is \textit{destructive} if it changes the value of $b$.
\end{definition}

\begin{definition}[Redundant-psync]
Consider an execution $E$ of a durable linearizable set. 
An explicit flush $f$ applied to the base object $b$, is redundant if $b$ was previously flushed by another explicit flush $f'$, and there does not exist a destructive write to $b$ between $f'$ and $f$ in $E$.
A persistence fence $f_p$ is redundant if there exists another persistence fence $f_p'$ prior to $f_p$ and there does not exist any non-redundant flush between $f_p'$ and $f_p$ in $E$.  
A psync which is a flush and a persistence fence, is redundant if the persistence fence redundant.
\end{definition}

\begin{theorem}
In an $n$-process system, for every durable linearizable lock-free set implementation $I$, there exists an execution of $I$ wherein $n$ processes are concurrently performing update operations and $n$-1 processes perform a redundant psync. 
\label{theorem:redundant-psync}
\end{theorem}
Izraelevitz et al. briefly mention that a helping mechanism for a non-blocking persistent object would include helping to persist operations \cite{izraelevitz2016linearizability}. 
Intuitively, durable linearizability requires that all operations that complete before a crash event are written to persistent memory.
If the response of update operation $\pi_1$ relies on the durability of a different operation $\pi_2$ then $\pi_1$  must ensure that the CPE of $\pi_2$ has occurred.
This is primarily a consequence of lock-free progress since $\pi_1$ cannot wait for other operations to ensure that the CPE of $\pi_2$ has occurred.
We can imagine an execution wherein $n$ processes all need to explicitly flush the same base object $b$ and perform a persistence fence but only one processes performs a destructive write on $b$.
Since only one destructive write was applied on $b$, only one explicit flush (and persistence fence) is necessary to guarantee that $b$ is written to persistent memory.
If all $n$ processes explicitly flush $b$ then all but one of the psyncs will be redundant. 
The simplest example of this execution would be $n$ processes concurrently executing $n$ identical update operations meaning those operations do not commute.
Suppose some process $p_1$ executing operation $\pi_1$ is the only process that successfully updates the set but sleeps after performing the CPE of the update operation.
The operations invoked by the other $n-1$ processes rely on the durability of $\pi_1$ thus, these $n-1$ process must also perform a psync to ensure that the CPE of $\pi_1$ has occurred since no process other than $p_1$ can determine if the CPE of $\pi_1$ has occurred. 
Suppose that each of the $n-1$ processes other than $p_1$ performs a psync on the base object involved in the CPE of $\pi_1$ then sleeps.
Then after every process has performed the psync all processes resume. 
In this case all $n-1$ processes other than $p_1$ have performed a redundant psync. 

We now describe the full proof construction.
\begin{proof}
Consider a durably linearizable lock-free set implementation $I$.
Assume for the purpose of showing a contradiction that no operation performs a redundant psync in every execution of $I$. 
Starting from an empty set, construct an execution $E$ of $I$ as follows:

Let $n$ processes each invoke identical \texttt{insert} operations such that process $p_i$ is performing the \texttt{insert} operation $\pi_i$ ($1 \leq i \leq n$).
As long as one of the processes continues to make progress, the decided response will be \texttt{true} for only one of the operations. 
Call this operation $\pi_s$.
Let $p_s$ run under no contention until immediately after the response of $\pi_s$ is decided.
Consider every other operation $\pi_i$, $\forall i \neq s$.
Since $\pi_i$ has the same arguments as $\pi_s$, if $p_i$ continues to make progress the decided response of $\pi_i$ will be \texttt{false}.
Let every $p_i$, $\forall i \neq s$, progress until immediately after $\pi_i$ is decided.

Consider the consequences of letting every $p_i$, $\forall i \neq s$ return from $\pi_i$ with a response value of \texttt{false} without confirming that the CPE of $\pi_s$ has occurred.
Suppose a crash event occurs immediately after the response of the operation that completed last.
Let $E'$ be the crash-recovery extension of $E$ and let $\pi'$ be identical to $\pi_s$.
Let $E''$ be $E' \cdot E_{\pi'}$ where $E_{\pi'}$ is the sequential execution beginning with the invocation of $\pi'$ and ending with the corresponding response.
Since the CPE of $\pi_s$ did not occur in $E$ and all $\pi_i$, $i \neq s$ failed in $E$, the response of $\pi'$ in $E''$ will be \texttt{true}.
Let $H''$ be the history of the crash-recovery extension of $E''$.
In this scenario all $\pi_i$, $i \neq s$ would have completed in $E$.
This means that $ops(H'')$ is not linearizable because the response of all $\pi_i$ reflects being linearized after $\pi'$ but $\pi'$ was invoked after the response of the $\pi_i$ that completed last.
This violates durable linearizability.
To avoid this problem, every process must confirm that the CPE of $\pi_s$ has occurred.

Let $b$ be the base object involved in the CPE of $\pi_s$.
Only process $p_s$ performed a destructive write on $b$.
This means that only one flush applied to $b$ will be non-redundant and therefore only one persistence fence following the flush will be redundant. 
Suppose one process continues until immediately after the CPE of $\pi_s$ then sleeps indefinitely.
If any other process flushes $b$ and performs a persistence fence then that process will have performed a redundant psync.
The CPE of $\pi_s$ does not change the volatile shared memory configuration.
This means that if the first process to progress past the CPE of $\pi_s$ sleeps indefinitely, no process will be able to determine if the CPE of $\pi_s$ has occurred.
Let $p_s$ continue until immediately after the CPE of $\pi_s$ then sleep indefinitely.
Every other process cannot complete until confirming that the CPE of $\pi_s$ has occurred.
However, since $p_s$ already explicitly flushed $b$ and performed a persistence fence any other process explicitly flushing $b$ and performing a persistence fence would violate our initial assumption but no process can make progress without flushing $b$ and performing a persistence fence. 
Thus we have reached a contradiction.
In this case $n$-1 processes $p_i$, $i \neq s$ must perform a redundant psync.
\end{proof}

\section{SLE linearizable sets and persistence-free reads
impossibility}
\label{section:lb-sle}

The key goal in Aguilera and Frølund was to enforce \textit{limited effect} by requiring operations to take effect before the crash or not at all. 
This is evident from the first paragraph of the abstract: "In such systems, an operation that crashes should either not happen or happen within some limited time frame—preferably before the process crashes. We define strict linearizability to achieve this semantics." 
The importance of limited effect is further emphasized in their own motivating example: “suppose that a military officer presses a button to launch a missile during war, but the missile does not come out. It might be catastrophic if the missile is suddenly launched years later after the war is over” \cite{AF03}. 
If the system can launch the missile during the execution of the recovery procedure (potentially many years later, when old systems are booted back up), this clearly contradicts the intentions of Aguilera and Frølund, however, in the full system crash-recovery model, strict linearizability admits this possibility.

% Aguilera and Frølund defined strict linearizability with the goal of enforcing \textit{limited effect} which requires that an operation takes effect within a limited amount of time after it is invoked.
Limited effect requires that an operation takes effect within a limited amount of time after it is invoked.
The point at which an operation takes effect is typically referred to as its \textit{linearization point} and we say that the operation \textit{linearizes} at that point.
Rephrasing the intuition, when crashes can occur, limited effect requires that operations that were pending at the time of a crash linearize prior to the crash or not at all.

\myparagraph{Strict linearizability does not satisfy limited effect in the full system crash-recovery model}
Strict linearizability is defined in terms of histories, which abstract away the real-time order of events.
As a result, strict linearizability does not allow one to argue anything about the ordering of linearization points of operations that were pending at the time of a crash relative to the crash event.
Thus, strict linearizability cannot and does not prevent operations from taking effect during the recovery procedure or even after the recovery procedure (which can occur for example in implementations that utilize linearization helping).
Strict linearizability only requires that at the time of a crash, pending operations \textit{appear} to take effect prior to the crash.
Although we are not aware of a formal proof of this, we conjecture in the full system crash-recovery model, durable linearizable objects are strict linearizable for some suitable definition of the recovery procedure.
This is because we can always have the recovery procedure \textit{clean-up} the state of the object returning it to a state such that the resulting history of any possible extension will satisfy strict linearizability.
We note this conjecture as further motivation towards re-examining the way in which the definition of strict linearizability has been adapted for the full system crash-recovery model.
The results in this paper do not rely on a proof of the conjecture.

\myparagraph{Limited effect matters}
It may be more accurate to say that strict linearizability enforces \textit{apparent limited effect}.
When the system can recover following a crash it is crucial that we distinguish between (1) requiring operations to actually \textit{take effect} before the crash and (2) requiring that operations simply \textit{appear to} take effect before the crash.
%This matters precisely because there is a difference between launching a missile years after 
% If this distinction were unimportant, then we should be satisfied with a recovery procedure launching a missile years after a conflict is over, as long as the operation that launches the missile is \textit{linearized} during the conflict. 

\myparagraph{Defining limited effect for sets}
In general arguing about an explicit point at which an operation takes effect is difficult. 
However, we are specifically interested in sets.
Update operations of a set are strongly non-commutative which is a useful property if we care about identifying an explicit point at which an update operation takes effect \cite{attiya2011laws}. 
% An operation $\pi_1$ is strongly non-commutative if there exists a configuration from which $\pi_1$ executed sequentially by process $p_1$ can alter the response of another operation $\pi_2$ executed sequentially by a different process $p_2$ \cite{attiya2011laws}.

In order to argue about the point at which an operation actually takes effect we need to identify specific event(s) at which operations can take effect.  
In our computational model, writes (or successful RMWs) are the only events with observable side effects since writes are the only events that change the states of shared objects.
Thus, for any implementation of a set, we argue that the linearization point of a successful update operation must be a write.
Specifically, the linearization point of a successful update operation is a \textit{key write} which we define as follows:
consider an execution $E$ of a durable linearizable set, ending with the write event $w$, on some base object $b$, as part of the successful update operation $\pi$, invoked by process $p$.
Let $E'$ be the prefix of $E$ ending just before $w$.
Consider a solo extension $EE$ of $E$ wherein some process $p'$ invokes and completes the operation $\pi'$ where $\pi'$ does not commute with $\pi$ and $\pi'$ is linearization help-free.
Similarly, consider the solo extension $EE'$ of $E'$ wherein process $p'$ invokes and completes the same operation $\pi'$ as previously defined.
The write event $w$ is the key write of $\pi$ iff the extensions $EE$ and $EE'$ exist and $\pi'$ has a different response in $EE$ and $EE'$. 
% 
% The write $w$ is the key write of $\pi$ iff there exists solo extensions $EE$ of $E$ and $EE'$ of $E'$ where in both $EE$ and $EE'$ process $p'$ invokes and completes the operation $\pi'$ where $\pi'$ does not commute with $\pi$ and $\pi'$ is linearization help-free and $\pi'$ has a different response in $EE$ and $EE'$. 

Intuitively, the key write is the event after which the operation can affect the response of some subsequent non-commutative, linearization help-free operation.
It is perhaps easier to imagine how a key write would exist for stacks or queues in which update operations typically contend on a single object (a top or head pointer) where it is clear that performing a write on that object will affect other operations. 

\myparagraph{Defining SLE linearizability}
We want to reason about when a key write occurs relative to a crash event. 
To do so we must lift the key write into histories. 
We utilize the concept of a key write to capture the intentions of Aguilera and Frølund in the context of sets by defining \textit{Strict limited effect} (SLE) linearizability for sets as follows:
an implementation $I$ of a set satisfies SLE linearizability iff all possible histories produced by $I$ satisfy SLE linearizability.
A history satisfies SLE linearizability iff the history satisfies strict linearizability and for all operations with a key write, if the operation is pending at the time of a crash, the key write of the operation must occur before the crash event.
In the strict completion of a history this is equivalent to requiring that the key write is always between the invocation and response of the operation. 
This is because the order of key writes relative to a crash event is fixed which means if the write occurs after the crash event then a strict completion of the history could insert a response for the operation only prior to the key write (at the crash) and this response cannot be reordered after the key write.
In this work we do not attempt to generalize SLE linearizability to other abstract types;
However, we note that it would not be difficult to generalize SLE linearizability for any type where we can define the concept of a key write for update operations.
For instance, \emph{stacks} and \emph{queues} are types where defining a key write is likely to be straightforward. 

We have already discussed some hand-crafted persistent sets which satisfy SLE linearizability.
In particular, the link-and-persist list of David et al. and the link-free list of Zuriel et al. both satisfy SLE linearizable.
% because these implementations ensure that the key write of any update operation always precedes the CPE of the operation. 

\myparagraph{Comparing with other correctness conditions}
SLE linearizability allows one to argue about the order of linearization points relative to crashes for sets.
SLE linearizability is stronger than durable linearizability and successfully captures the intuition that motivated strict linearizability.
One might say that SLE linearizability is similar to \textit{detectable executions} defined by Friedman et al. \cite{QueueFriedman18} in that ensuring SLE linearizability will often require identifying operations that were pending at the time of a crash in order to revert any effects of those pending operations whose key write did not occur prior to the crash. 
If the implementation is capable of identifying operations that were pending at the time of a crash then it provides detectable execution. 

\subsection{SLE Linearizable sets and persistence-free reads impossibility result}

We show that it is impossible to implement a SLE linearizable lock-free set for which read-only searches do no perform any explicit flushes or persistence fences. 
This result primarily follows from the fact that some operations will rely on the durability of other operations but persistence events have no visible side effects meaning from the perspective of a persistence-free search, a configuration in which the CPE of a particular operation has occurred is not distinguishable from another configuration in which the CPE of said operation has occurred.
An example is depicted in \Cref{fig:lb1-executions}.

\begin{figure}[!t]
\begin{center}
    \includegraphics[width=0.6\linewidth]{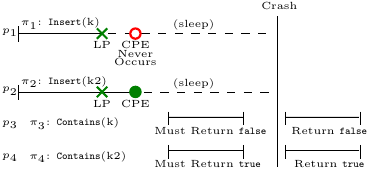}
\caption{A simplified example of the execution described in \cref{theorem:persistence-free-search-impossiblity}. 
Processes $p_2$ and $p_3$ cannot distinguish between configurations resulting after the last step of $p_1$ or $p_2$. 
$p_3$ and $p_4$ have no way to determine if the CPE of $\pi_1$ and $\pi_2$, respectively, have occurred.
$\pi_3$ and $\pi_4$ must return a specific value since an identical operation invoked after the crash has a fixed response but this response is not know a priori.}
\label{fig:lb1-executions}
\end{center}
\end{figure}

\begin{definition}[Persistence-Free Searches]
We say that a set offers persistence free searches if there is no execution of the set in which a process that invoked a read-only operation $\pi$ performs an explicit flush or persistence fence between the invocation and response of $\pi$. 
\end{definition}

\begin{lemma}
A SLE linearizable lock-free set with read-only searches cannot linearize successful update operations after the CPE of the update.
\label{lemma:sle-lin-constraint}
\end{lemma}

\begin{proof}
Let $E$ be an execution of a SLE linearizable lock-free set $I$ in which a single process $p$ runs alone, and invokes an update operation $\pi$.
Let the initial configuration $c$ be constructed such that $\pi$ will eventually return \texttt{true} assuming $p$ continues to make progress.
This means that $\pi$ has a key write $w$.
Assume that $\pi$ is linearized at some point after its CPE.
Let $p$ run until immediately after the CPE of $\pi$.
Suppose a crash event occurs.
This means that $w \not\in E$.
Let $E'$ be the crash-recovery extension of $E$.
Let $\pi'$ be an operation identical to $\pi$. 
Let $E''$ be the $E' \cdot E_{\pi'}$ where $E_{\pi'}$ is the sequential execution of the operation $\pi'$ by some process.
Since $\pi$ has a CPE in $E$, $\pi'$ will return \texttt{false} in $E''$.
$\pi$ is the only pending operation in $E'$.
Let $H$ be the history of $E'$ and $H^c$ be the strict completion of $H$.
For simplicity we omit the invocation and response of the recovery procedure from $H^c$.
There are two cases for the constructing $H^c$:

Case 1: The invocation of $\pi$ is removed meaning $H^c$ contains only the invocation and response of $\pi'$ where the response of $\pi'$ was \texttt{false}.
There is no possible sequential execution of $I$ beginning from $c$ in which only one operation, $\pi'$ is invoked and returns \texttt{false}.
This contradicts $H^c$ in which $\pi'$ returns \texttt{false}.

Case 2: A response for $\pi$ is inserted meaning $H^c$ contains the invocation and response of $\pi$ followed by the invocation and response of $\pi'$.
In a sequential execution of $I$ beginning from $c$, if $\pi$ is invoked first then $\pi$ will return \texttt{true} and $\pi'$ will return false.
This means that the key write of $\pi$ has occurred.
However, $w$, the key write of $\pi$, did not occur prior to the crash because $w \not\in E$.
Thus this is a contradiction.
Since both cases violate SLE linearizability, this proves that sets that satisfy SLE linearizability cannot linearize update operations after the CPE of the update.
\end{proof} 

Intuitively, Lemma 9 captures the fact that for lock-free sets with read-only searches an update operation cannot be linearized after its CPE because this could result in the key write of the operation being performed after a crash event.
% which directly contradicts the guarantees provided by SLE linearizability.
Recall that the key write of a successful update is the same as the linearization point of the update.
The proof of Lemma 9 is by counterexample: consider an execution wherein a crash occurs immediately after the CPE of a successful update operation $\pi$.
Since $\pi$ performed its CPE this means that an identical update invoked after the crash-recovery extension will be unsuccessful which is only possible if the key write of $\pi$ was performed;
however, since the crash occurred prior to the linearization point of $\pi$ this means that that key write of $\pi$ must have occurred after the crash.
% A full version of the proof is presented in \Cref{appendix:proofs-lb1}.

% \setcounter{theorem}{1}
\begin{theorem}
There exists no SLE linearizable lock-free set with persistence-free read-only searches.
\label{theorem:persistence-free-search-impossiblity}
\end{theorem}

\begin{proof}
Consider a SLE linearizable lock-free set implementation $I$.
Assume for the purpose of showing a contradiction that every search is persistence-free in every execution of $I$. 
Starting from an empty set, construct an execution $E$ of $I$ as follows:
Let process $p_u$ invoke an \texttt{insert} operation to insert the key $k$ then let $p_u$ progress until immediately prior to the CPE of $\pi_u$ then sleep.
By lemma \ref{lemma:sle-lin-constraint} the linearization point of $\pi_u$ must be at or before its CPE.
Consider a search operation $\pi_c$ looking for the key $k$.
This means that $\pi_c$ does not commute with $\pi_u$.

Case 1: $\pi_u$ is linearized prior to its CPE.
Suppose process $p_c$ invokes the search operation $\pi_c$.
Since the next event of $\pi_u$ would be its CPE and it was linearized prior to its CPE this means that the linearization point of $\pi_u$ occurred before the invocation of $\pi_c$.
Let $p_c$ complete $\pi_c$ in $E$.
Since $\pi_u$ was linearized the response of $\pi_c$ must return \texttt{true}.
Now suppose that a crash event occurs.
Let $E'$ be the crash-recovery extension of $E$ and let $\pi_c'$ be a search operation identical to $\pi_c$.
Let $E''$ be a solo extension of $E'$ by any process in which $\pi_c'$ is invoked and completes.
Since the CPE of $\pi_u$ never occurred in $E$ (or $E'$) this means that $\pi_c'$ will return false in $E''$.
Let $H_1$ be the history of $E''$ and let $H_1^c$ be the strict completion of $H_1$.
There is no sequential execution of $I$ that could produce a history equivalent to $H_1^c$ for any construction of $H_1^c$.
Consider a sequential execution $E_s$ wherein the only operations invoked are two identical search operations and both searches complete in $E_s$.
Since no updates were linearized between the searches, the response must be the same for both search operations in $E_s$.
This is true for any execution $E_x \cdot E_s$.
Since $H_1^c$ has no equivalent sequential history SLE linearizability is violated.

% ------------------------------------------------------------
% ------------------------------------------------------------
% ------------------------------------------------------------
Case 2: $\pi_u$ is linearized at its CPE.
Since the CPE of an operation is a persistence event and persistence events have no effect on volatile shared memory $\pi_c$ cannot determine if the CPE of $\pi_u$ has occurred.
If a process invokes the search operation $\pi_c$ defined in case 1, assuming $p_c$ continues to progress $\pi_c$ must eventually return.
The following cases describe two different indistinguishable executions and demonstrate why $\pi_c$ cannot arbitrarily return \texttt{true} or arbitrarily return \texttt{false}.

Case 2A: To avoid confusion with case 1 we will let the execution $E_A$ be equivalent to $E$.
Let $p_c$ invoke and complete $\pi_c$ in $E_A$.
$\pi_c$ observes the effects of $\pi_u$ and assumes that the CPE of $\pi_u$ has occurred so it returns \texttt{true}.
Now suppose a crash event occurs. 
Let $E_A'$ be the crash-recovery extension of $E_A$ and let $E_A''$ be the solo extension of $E_A'$ by any process in which $\pi_c'$ is invoked and completes.
The response of $\pi_c'$ in $E_A''$ will be \texttt{false} since the CPE of $\pi_u$ did not occur.
This is now equivalent to case 1. 
The strict completion of the resulting history has no equivalent sequential history so SLE linearizability is violated.
% %%%%%%%%%%

Case 2B: Let $E_B$ be the extension of $E$ where $p_u$ progresses until immediately after the CPE of $\pi_u$ then sleeps indefinitely.
Now let $p_c$ invoke and complete $\pi_c$.
The only difference between $E_A$ and $E_B$ is the fact that the CPE of $\pi_u$ occurs in $E_B$.
Since these executions are indistinguishable and the correct response of $\pi_c$ in case 2A should have been \texttt{false} we will let the response of $\pi_c$ in $E_B$ be \texttt{false}.
Now suppose a crash event occurs. 
Let $E_B'$ be the crash-recovery extension of $E_B$ and let $E_B''$ be the solo extension of $E_B'$ by any process in which $\pi_c'$ is invoked and completes.
The response of $\pi_c'$ in $E_B''$ will be \texttt{true} since the CPE of $\pi_u$ did not occur.
As in the previous cases the strict completion of the resulting history has no equivalent sequential execution since the update was linearized before the searches were invoked both searches should have the same response.
% %%%%%%%%%%

The executions described in case 1 and case 2 violate SLE linearizability. 
In all cases, SLE linearizability would not be violated if $\pi_c$ completed the CPE of $\pi_u$ and returned a value reflecting that the update was linearized.
However, if $\pi_c$ does complete the CPE of $\pi_u$ this requires performing at least one persistence event which would contradict our assumption.
Alternatively $p_c$ could wait for some other process to complete the CPE of $\pi_c$ and somehow signal the fact that the CPE has completed however this contradicts lock-freedom.
It is trivial to construct equivalent executions for these cases where the update operation is a \texttt{remove} and the search is looking for the key removed by the update. 
Thus there exists no SLE linearizable lock-free set with persistence-free search operations.
\end{proof}

\Cref{theorem:persistence-free-search-impossiblity} follows from Lemma 9 combined with the fact that persistence-fences and background flushes have no visible side effects.
% A formal proof of \cref{theorem:persistence-free-search-impossiblity} is provided in the \Cref{appendix:proofs-lb1}.
% The proof describes the existence of an execution in which an \texttt{Insert} operation is concurrent with a persistence-free search looking for the key being inserted and the search is forced to perform a psync.
For any lock-free SLE linearizable set we can construct an execution in which an \texttt{Insert} operation is concurrent with a persistence-free search looking for the key being inserted and the search is forced to perform a psync.
If the search observes the effects of the update it must be able to complete due to lock-freedom.
However, if it returns \texttt{true} but the CPE of the update never occurs before a crash then the resulting history is not SLE linearizable.
The search cannot arbitrarily assume that the CPE has not occurred since the CPE could be a background flush that occurred before the invocation of the search so if the search completes and arbitrarily returns \texttt{false} then the resulting history would still not be SLE linearizable.
These cases are indistinguishable since the in both cases the search must be able to determine if the CPE of the update has occurred but this is not possible since persistence events have no visible side effects. 
As a result the search must either perform a psync or wait a possibly infinite amount of time. 
To allow for persistence-free reads, one must either sacrifice SLE linearizability or allow blocking.

\fi

\section{Upper Bounds}
\label{section:algos}
\myparagraph{Briding the gap between theory and practice}
The lower bounds presented in the previous section offer insights into the theoretical limits of persistent sets for both durable linearizability and SLE linearizability.
While these lower bounds demonstrate a clear separation between durable and SLE linearizability, it is unclear whether or not we can observe any meaningful separation in practice.
In order to answer this question we would like to compare durable and SLE linearizable variants of the same persistent set implementation.
To this end, we extended the \LPList~technique \cite{david2018log} to allow for persistence-free searches and use our extension to implement several persistent linked-list.
We also add persistence helping to \SOFTList~ \cite{zuriel2019efficient}. We explain both in detail next. 

\myparagraph{Notable persistent set implementations}
We briefly describe above mentioned existing implementations of persistent sets. 
We only focus on hand-crafted implementations since they generally perform better in practice compared to transforms or universal constructions~\cite{friedman2020nvtraverse,friedman2021mirror}.

David et al. describe a technique for implementing durable linearizable link-based data structures called the \LPList~technique \cite{david2018log}.
Using the \LPList~technique, whenever a link in the data structure is updated, a single bit mark is applied to the link which denotes that it has not been written to persistent memory.
The mark is removed after the link is written to persistent memory.
We refer to this mark as the \textit{persistence bit}.
This technique was also presented by Wang et al. in the same year \cite{wang2018easy}.
Wei et al. presented a more general technique which does not steal bits from data structure links \cite{wei2021flit}.

The \LFList~algorithm of Zuriel et al. does not persist data structure links \cite{zuriel2019efficient}. 
Instead, the \LFList~algorithm persists metadata added to every node. 
 
Zuriel et al. designed a different algorithm called \SOFTList~(Sets with an Optimal Flushing Technique) offering persistence-free searches. 
The \SOFTList~algorithm does not persist data structure links and instead persists metadata added to each node. 
The major difference between the \LFList~algorithm and \SOFTList~is that \SOFTList~uses two different representations for every key in the data structure where only one representation is is explicitly flushed to persistent memory. 

\myparagraph{Recovery complexity}
After a crash, a recovery procedure is invoked to return the objects in persistent memory back to a consistent state.
Prior work has utilized a sequential recovery procedure \cite{zuriel2019efficient,david2018log,friedman2021mirror,correia2020persistent}. 
A sequential recovery procedure is not required for correctness but it motivates the desire for efficient recovery procedures.
No new data structure operations can be invoked until the recovery procedure has completed. 
Ideally we would like to minimize this period of downtime represented by the execution of recovery procedure. 
For the upper bounds in the this section, we use the asymptotic time complexity of the recovery procedure as another metric for comparing durable linearizable algorithms.

\myparagraph{Extended Link-and-Persist}
We choose to extend the \LPList~technique of David et al. because it is quite simple and it represents the state of the art for hand-crafted algorithms that persist the links of a data structure.
Moreover, unlike the algorithms in \cite{zuriel2019efficient}, the \LPList~technique can be used to implement persistent sets without compromising recovery complexity.
We build on the \LPList~technique by extending it to allow for persistence-free searches and improved practical performance.
Cohen et al noted that persistence-free searches rely on the ability to linearize successful update operations at some point after the CPE of the operation \cite{cohen2018inherent}. 
In our case, this means that searches must be able to determine if the pointer is not persistent because of an \texttt{Insert} operation or a \texttt{Remove} operation.
This is not possible with the original \LPList~technique.
We address this with two changes.

First, we require that a successful update operation, $\pi_u$, is linearized after its \emph{Critical Persistence Event} (or CPE).
Intuitively, the CPE represents the point after which the update will be recovered if a crash occurs.
Specifically, if a volatile data structure would linearize $\pi_u$ at the success of a RMW on a pointer $v$ then we require that $\pi_u$ is linearized at the success of the RMW that sets the persistence bit in $v$.
If a search traverses a pointer, $v$, marked as not persistent the search can always be linearized prior to the concurrent update which modified $v$. 

Secondly, since successful updates are linearized after their CPE, if the response of search operation depends on data that is linked into the data structure by a pointer marked as not persistent then the search must be able to access the last persistent value of that pointer.
To achieve this, we add a pointer field to every node which we call the \textit{old field}.
A node will have both an \textit{old field} and a pointer to its successor (\textit{next pointer}) which effectively doubles the size of every data structure link.
The \textit{old field} will point to the last persistent value of the successor pointer while the successor pointer is marked as not persistent.
In practice, the \textit{old field} must be initialized to \texttt{null} then updated to a non-\texttt{null} value when the corresponding successor pointer is modified to a new value that needs to be persisted.
Note that modifications like flagging or marking do not always need to be persisted; this depends on the whether or not the update can complete while the flagged or marked pointers are still reachable via a traversal from the root of the data structure.
The easiest way to correctly update the \textit{old field} is to update the successor pointer and the \textit{old field} atomically using a hardware implementation of double-wide compare-and-swap (DWCAS) namely the cmpxchg16b instruction on Intel.
Alternatively, a regular single-word compare-and-swap (SWCAS) can be used but this requires adding extra volatile memory synchronization to ensure correctness.
For some data structures such as linked-lists using only SWCAS might also require adding an extra psync to updates. 
To allow searches to distinguish between pointers that are marked as not persistent because of a \texttt{remove} versus those that are not persistent because of an \texttt{insert} we require that the \textit{old field} is always updated to a non-\texttt{null} value whenever a \texttt{remove} operation unlinks a node.
\texttt{Insert} operations that modify the data structure must flag either the \textit{old field} or the corresponding successor to indicate that the pointer marked as not persistent was last updated by an insert.
When using SWCAS to update the \textit{old field} this flag must be on the successor pointer.

With our extension if the response of a search operation depends on data linked into the data structure via a pointer marked as not persistent it can be linearized prior to the concurrent update operation that modified the pointer and it can use the information in the \textit{old field} to determine the correct response which does not require performing any psyncs.
If the search finds that the update was an insert it simply returns \texttt{false}.
If the update was a remove but the search was able to find the value that it was looking for then it can return \texttt{true} since that key will be in persistent memory. 
If the update was a remove but the search was not able to find the value that it was looking for then it can check the if the node pointed to by the \texttt{old field} contains the value.
As with the original, our extension still requires that an operation $\pi$ will ensure that the CPE of any other operation which $\pi$ depends on has occurred. 
$\pi$ must also ensure that its own CPE has occurred before it returns. 
Another requirement which was not explicitly stated by David et al. is that operations must ensure that any data that a data structure link can point to is written to persistent memory before the link is updated to point to that data.

Our extension can be used to implement several link-based sets including trees and hash tables.
Data structures implemented using our extension provide durable linearizability, however the use of persistence-free searches is optional.
If the data structure does not utilize persistence-free searches then it would provide SLE linearizability (requiring only a change in the correctness proof).

\myparagraph{Augmenting LF and SOFT}
\SOFTList~ represents the state of the art for hand-crafted algorithms that do not persist the links of a data structure.
The \SOFTList~ algorithm provides durable linearizability. 
For comparison, we added persistence helping for all operations of a persistent linked-list implemented using \SOFTList~ (thereby removing persistence-free searches) to achieve a SLE linearizable variant.
We refer to this variant as SOFT-SLE.
We also modified the implementation of the \LFList~ algorithm. 
While the original \LFList~ algorithm does not explicitly persist data structure links, it still allocates the links from persistent memory.
We can achieve better performance by allocating the links from volatile memory.
To emphasize the difference we refer to this as LF-V2.

\subsection{Our Persistent List Implementations}
In order to compare our extension to existing work we provide several implementations of persistent linked-lists which utilize our extended-link-and-persist approach.
We choose to implement and study linked lists because they generally do not require complicated volatile synchronization. 

\begin{lstlisting}[frame=bt, label={fig:pf-search-code}, float=t!, caption={Pseudocode for the persistence-free contains function of our Physical-Delete (PD) list. The volatile synchronization is based on the list of Fomitchev and Ruppert.}]
def PersistenceFreeContains(key) :
    p = head, pNext = p.next, curr = UnmarkPtr(pNext)
    while true :
        if curr.key $\leq$ key : break
        p = curr, pNext = p.next
        curr = UnmarkPtr(pNext)
    hasKey = curr.key==key
    if IsDurable(ptNext) : return hasKey //\label{line:physical-del-pf-contains-ret1}
    old1 = p.old, pNext2 = p.next, old2 = p.old //\label{line:physical-del-pf-contains-reads-start}
    pDiff = pNext$\neq$pNext2, oldDiff = old1$\neq$old2 //\label{line:physical-del-pf-contains-reads-end}
    if pDiff or oldDiff or old1==null : return hasKey //\label{line:physical-del-pf-contains-ret2}
    if IsIFlagged(old1) : return false //\label{line:physical-del-pf-contains-ret3}
    if hasKey : return true //\label{line:physical-del-pf-contains-ret4}
    return UnmarkPtr(old1).key==key//\label{line:physical-del-pf-contains-ret5}
\end{lstlisting}

We refer to our implementations as PD (Physical-Delete), PD-S (SWCAS implementation of PD), LD (Logical-Delete) and LD-S (SWCAS implementation of LD).
The names refer to the synchronization approach and primitive.
Our implementations use two different methods for achieving synchronization in volatile memory.
Specifically we use one based on the Harris list \cite{harris-set} and another based on the work of Fomitchev and Ruppert \cite{fomitchev2004lock}.
The former takes a lazy approach to deletion that relies on marking for logical deletion and helping.
As a result, marked pointers must be written to persistent memory which requires an extra psync.
The latter does not take a lazy approach to deletions but still relies on helping and requires extra volatile memory synchronization through the use of marking and flagging. 
Fortunately, we do not need to persist marked or flagged pointers with this approach.
\Cref{fig:list-diagram} shows an example of an update operation in the PD list implementation. 
We also implement separate variants using 2 different synchronization primitives, DWCAS and SWCAS.
\Cref{table:durable-lists-updates} summarizes some of the details of these approaches. 
We assume that the size of the \textit{key} and \textit{value} fields allow a single node to fit on one cache line meaning a \textit{flush} on any field of the node guarantees that all fields are written to persistent memory.
The assumption that the data we want to persist fits on a single cache line is common.
David et al., Zuriel et al. and several others have relied on similar assumptions \cite{zuriel2019efficient,david2018log,correia2020persistent,ramalhete2021efficient}.
It is possible that our persistent list could be modified to allow for the case where nodes do not fit onto a single cache line by adopting a strategy similar to \cite{cohen2017efficient}.

\begin{figure}[!t]
\begin{center}
    \includegraphics[width=0.5\linewidth]{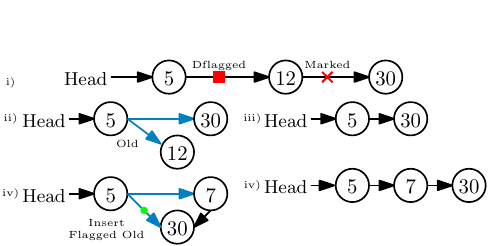}
\caption{Steps to execute an \texttt{insert(7)} operation in our PD list implementation. Blue pointers indicate non-durable pointers (with persistence bits set to 0). i) Initially we have three nodes. The node containing 5 has a pending delete flag (Dflagged) and the node containing 12 is marked for deletion. We traverse to find a key $\geq$ 7. ii) Help finish the pending \texttt{Remove} via DWCAS to unlink marked node and set old pointer. iii) Flush and set persistence bit via DWCAS (clearing old pointer). iv) Via DWCAS insert 7 and set old pointer. The old pointer is flagged to indicate a pending insert. v) Flush and set persistence bit via DWCAS.} 
\label{fig:list-diagram}
\end{center}
\end{figure}

\myparagraph{Search Variants}
As part of our persistent list, we implement 4 variants of the \texttt{contains} operation: persist all, asynchronous persist all, persist last and persistence free.
We focus on the latter two since the others are naive approaches that perform many redundant psyncs.

\myparagraph{Persist Last (PL)} 
If the pointer into the terminal node of the traversal is marked as not persistent then write it to persistent memory and set its persistence bit via a CAS.
This variant is the most similar to the searches in the linked list implemented using the original \LPList~technique.

\myparagraph{Persistence Free (PF)} 
If the pointer into the terminal node of the traversal performed by the search is marked as not persistent then use the information in the \textit{old field} of the node's predecessor to determine the correct return value without performing any persistence events.
Since we do not need to set the durability bit of any link, this variant does not perform any writes and never performs a psync.
Algorithm \ref{fig:pf-search-code} shows the pseudocode for the persistence-free search of the PD list. 
For simplicity we abbreviate some of the bitwise operations with named functions.
Specifically, \textit{UnmarkPtr} which removes any marks or flags, \textit{IsDurable} which checks if the pointer is marked as persistent and \textit{IsIflagged} which checks if the pointer was flagged by an \texttt{insert}.

\begin{theorem}
The PD, PD-S, LD, and LD-S lists are durable linearizable and lock-free. 
\label{list-correctness}
\end{theorem}

We prove \Cref{list-correctness} in the full version of the paper. We can also show that our list implementations are durable linearizable by considering a volatile abstract set (the keys in the list that are reachable in volatile memory) and a persistent abstract set (the keys in the list that are reachable in persistent memory).
By identifying, for each operation, the points at which these sets change, we can show that updates change the volatile abstract set prior to changing the persistent abstract set and that each update changes the the volatile abstract set exactly once.
It follows that the list is always consistent with some persistent abstract set.

If we never invoke a persistence-free \texttt{contains} operation then we can prove that the implementations are SLE linearizable and lock-free. 
Doing so simply requires that we change our arguments regarding when we linearize update operations such that the linearization point is not after the CPE.
Note that of the set implementations that we discuss, those that have persistence-free searches are examples of implementations which are strict linearizable but not SLE linearizable. 
These implementations require that the recovery procedure or operations invoked after a crash take steps which effectively linearize operations.
This is because following a crash, one cannot tell the difference between an operation that has progressed far enough to allow some future operation to help linearize it and an operation that was already linearized.

\begin{table}[t]
\centering
\begin{tabular}{|l|l|l|l|l|}
\hline
\textbf{\begin{tabular}[c]{@{}l@{}}Name\end{tabular}} &
\textbf{\begin{tabular}[c]{@{}l@{}}Synch. Approach\end{tabular}} &
\textbf{\begin{tabular}[c]{@{}l@{}}Synch. Primitive\end{tabular}} &
\multicolumn{2}{l|}{\textbf{\begin{tabular}[c]{@{}l@{}}Min Psyncs Per Insert/Remove\end{tabular}}} \\ \hline
PD              & Fomitchev   & DWCAS  & 1 \hspace{4mm} & 1 \\ \hline
PD-S            & Fomitchev   & SWCAS    & 2 & 1 \\ \hline
LD              & Harris      & DWCAS  & 1 & 2 \\ \hline
LD-S            & Harris      & SWCAS    & 2 & 2 \\ \hline
\end{tabular}
\caption{Our Novel Persistent List Details.}
\label{table:durable-lists-updates}
\vspace{-1cm}
\end{table}

\ifx\fullpaper\undefined
\else
\section{Persistent List Detailed Description}
\label{section:durable-list-code}

Our approach to implementing persistent sets is an extension of the \LPList~technique \cite{david2018log}.
One major problem with the \LPList~technique is that it cannot be used to implement persistence free \texttt{contains} operations.
This is primarily because the \LPList~technique offers no way to determine what type of update operation caused a link to become marked as not persistent. 

To address this issue, we need a mechanism that allows read-only operations to distinguish between links that are not persistent because of an incomplete \texttt{insert} operation and links that are not persistent because of an incomplete \texttt{remove} operation.
We must also be able to safely linearize update operations after the critical persistence step. 
The latter constraint is already possible with the original \LPList~technique, however the former is less trivial. To understand the challenges consider an execution with two processes, $p_1$ and $p_2$ where $p_1$ is performing a \texttt{contains} operation $\pi_1$ concurrently with $p_2$ which is performing a \texttt{remove} operation $\pi_2$. 
Assume that $\pi_1$ is persistence-free.
Let $k$ be the key that $\pi_1$ is searching for and let $k$ also be the key that $\pi_2$ wants to remove.
Suppose that $k$ is already in the set and is contained in the node $n$. 
Let $p_2$ progress until it unlinks $n$ then let $p_2$ sleep indefinitely at some point before the CPE of $\pi_2$. 
$p_1$ will eventually traverse the link updated as part of $\pi_2$ which will lead to the successor of $n$.
The \textit{persistence bit} of the link will indicate that it is has not been written to persistent memory.
% The response of $\pi_1$ must reflect whether or not $\pi_2$ is linearized.
The \textit{persistence bit} cannot alone be used to determine the correct response for $\pi_1$.
% The read-only operation could not be persistence free if update operations are linearized at or before the critical flush of the update.
% If updates were linearized at or before the critical flush of the update then a concurrent read-only operation would have no way to deterimine if the update has been linearized or not.
If a system crash occurs the contents of persistent memory will reflect that $k$ is still the data structure.
$\pi_1$ is persistence-free so it will not complete the CPE of $\pi_2$ and we cannot assume that the CPE of $\pi_2$ will occur as a result of a background flush.
This means that $\pi_1$ must either perform a persistence event or it must be linearized prior to $\pi_2$.
The former violates our assumption.
The latter requires that $\pi_2$ is linearized after its CPE.
Since we are not concerned with guaranteeing SLE linearizability, we can assume that $\pi_2$ is linearized after its CPE.
This means that $p_1$ must be able to determine the key contained in the node that $\pi_2$ unlinked.
% This is required because we have no guarantee that the link marked as not durable has been written to persistent memory and the unlinked node contains $k$.
However, in volatile shared memory, $n$ is not reachable via a traversal from the root.
Since $p_1$ has no way to determine the key contained in $n$, if $\pi_1$ returns \texttt{false} the execution is not durable linearizable.  
For the case where $\pi_1$ returns true simply reconstruct the example starting from a configuration wherein the node unlinked by $\pi_1$ does not contain $k$ and $k$ is not contained in any node in the list.
This means that there is no safe way to linearize $\pi_1$ without performing a persistence event.

To solve this problem, we add a field to every node in the data structure.
We refer to this field as the \textit{old field} (or \textit{old pointer}).
The \textit{old field} is a pointer to a node and is always initialized to \texttt{null} when a new node is created.
When a remove operation $\pi$ wants to unlink a node $n$ by atomically updating the next pointer of its predecessor $u$, it must also atomically update the \textit{old field} of $u$ to point to $n$. 
This means that $n$ will remain accessible until some point after the CPE of $\pi$.
After the CPE of $\pi$ has occurred the \textit{old field} in $u$ can be reverted back to \texttt{null}.
We also utilize the old field in insert operations to achieve a minimum of 1 psync per insert operation.
When an \texttt{insert} operation wants to insert a new node by updating the next point of some node $n$ it will first set the \textit{old field} in $n$ to the last persistent successor of $n$.
A mark bit is applied to the \textit{old field} when the value is set by an \texttt{insert} operation.
We refer to this as the \textit{insert bit}.
This allows operations to easily determine the type of operation that set the \textit{old field}.

If the \textit{old field} in a node is not \texttt{null} then there is an ongoing concurrent update operation involving the node.
Update operations can use the \textit{old field} to either fail early or help concurrent operations.
\texttt{Contains} operations can use the old field to avoid performing persistence events.

\subsection{SWCAS vs. DWCAS}
The way that we utilize the old field allows us to think of data structure links as a tuple in the form $\langle next, old\rangle$.
Moreover, in a real implementation, both $next$ and $old$ are word sized pointers.
DWCAS is therefore a natural choice if we want to atomically update $\langle next, old\rangle$.
On the other hand, there are some performance concerns related to the use of DWCAS.
Not all processors support hardware implementations of DWCAS and software implementations are typically inferior. 
% such as the one used in the C++ std::atomic library can be major performance bottlenecks (DWCAS in the C++ std::atomic.
This suggested that a SWCAS implementation would perform better.
For this reason, we implement different versions of our persistent list using both DWCAS and SWCAS respectively.
% We compare the DWCAS implementations implementations of our list to the SWCAS implementations in \cref{section:swcas-vs-dwcas}.

\subsection{Pseudo Code Overview}
In this section we will provide an more in depth description of the different implementations of our durable list.
We will focus primarily on the DWCAS implementation of the \textit{Physical-Del} list since this is version performed best in the experimental tests.
% We will also briefly describe the DWCAS version of the \textit{Logical-Del} list and the corresponding SWCAS version.
Throughout the pseudo code we utilize the functions \texttt{UnmarkPtr, IsDurable, MarkDurable} and others of the form \textit{IsX} or \textit{MarkX}.
These functions represent simple bitwise operations used to remove marks/flags, check if a node is marked/flagged and apply marks/flags.
We omit their full function bodies. 
We also use the function \texttt{Flush}.
On Intel systems \texttt{Flush} can be a single \textit{clflush} instruction, a \textit{clflushopt} and an \textit{sfence} or a \textit{clwb} and an \textit{sfence}.

Every version of our persistent list uses the same \textit{Node} data type.
A \textit{Node} contains four fields: \textit{key, value, next} and \textit{old}.
\textit{Next} and \textit{old} are word sized pointers.
We assume that the size of the \textit{key} and \textit{value} fields allow a single Node to fit on one cache line meaning a \textit{flush} on any field of the node guarantees that all fields are written to persistent memory.
Nodes represent the \textit{critical} data that we want to persist.
The assumption that the critical data fits on a single cache line is common.
The \LPList~list, \LFList~list and \SOFTList~list require similar assumptions.
It is possible that our persistent list could be modified to allow for the case where nodes do not fit onto a single cache line by adopting a strategy similar to \cite{cohen2017efficient}.

Every version of our persistent list utilizes a dummy \textit{head} and \textit{tail} node where the root of the list is always head, head has no predecessor and contains the key equivalent to negative infinity. 
Tail contains the key equivalent to infinity and has no successor.
An empty list consists of only the head and tail nodes where $head.next$ points to $tail$.

\subsubsection{Physical-Del}

%%%--------------------------------
%%          Search + Persist 
%%%--------------------------------
\begin{lstlisting}[frame=single, float=t!, label={fig:helpers1}, caption={Persist and Find helper functions for Physical-Delete (DWCAS)}]
def Persist(node, expNext, old) :
    Flush(&node.next)
    durNext = MarkDurable}(expNext)
    if old == null :
        old = node.old
    des = $\langle$node.next, node.old$\rangle$
    exp = $\langle$expNext, old$\rangle$
    new = $\langle$durNext, null$\rangle$
    DWCAS(des, exp, new) //\label{line:physical-del-persist-dwcas}   
    
def Find(key) :
    gp = null
    p = head
    pNext = p.next
    curr = UnmarkPtr(pNext)
    while true :
        if curr.key $\geq$ key :
            break
        gp = p
        p = curr
        pNext = p.next
        curr = UnmarkPtr(pNext)
    if gp $\neq$ null :
        gpNext = gp.next
        if not IsDurable(gpNext) :
            Persist(gp, gpNext, null) //\label{line:search-dur-bit-gp}
    if not IsDurable(pNext) :
        Persist(p, pNext, null) //\label{line:search-dur-bit-parent}
    return $\langle$gp,p,curr$\rangle$
\end{lstlisting}

Algorithms \ref{fig:pf-search-code}-\ref{fig:physical-del-dwcas-helpers2} show the pseudo code for the \physdelList~list which uses DWCAS.
This implementation is based on the volatile list of Fomitchev and Ruppert \cite{fomitchev2004lock}.

With this approach, we utilize the bottom three bits of a node's next pointer for flagging and marking. 
A link in the data structure can be represented as the tuple: $\langle$next, dflag, marked, persistence-bit$\rangle$.
The bottom bit is the persistence bit which as in \LPList-technique indicates if the link is persistent.
However, in our implementations if the persistence bit is set then the node is persistent meaning the persistence bit is applied to a link after it is flushed.
The \textit{marked} bit is used to mark a node as logically deleted.
The \textit{dflag} bit indicates that a remove operation is going to remove the successor of the node.
Each of these marks are initialized to 0. 

This implementation uses DWCAS.
Whenever a node is updated both its \textit{next} and \textit{old} fields are updated atomically.
Any update that sets the persistence bit in the next pointer of a node will also revert the old field to \texttt{null}. 

We will now give an overview of the main functions utilized by the \physdelList~list.
Many of these functions will be very similar to the volatile list in \cite{fomitchev2004lock} with the main differences being related to achieving a persistence-free search.

\myparagraph{Find} 
The \texttt{Find} function is shown in Algorithm \ref{fig:helpers1}. 
\texttt{Find} has a single argument, $key$ and returns the tuple $\langle gp, p, curr\rangle$ where $gp$, $p$ and $curr$ are nodes such that $curr.key \geq key$ and there was a time during the execution of the function where $gp.next$ pointed to $p$ and $p.next$ pointed to $curr$.
As in the \LPList-technique, \texttt{Find} will also confirm that $gp.next$ and $p.next$ are persistent.

\begin{lstlisting}[name=updates,frame=single, float=t!, label={fig:physical-del-dwcas-updates1}, caption={Insert function for Physical-Delete (DWCAS)}]
def Insert(key, val) :
    while true :
        $\langle$gp,p,curr$\rangle$ = Find(key) //\label{line:insert-search}
        pNext = p.next
        if curr.key == key : 
            return false
        if not IsClean(pNext) :
            HelpUpdate(gp, p)//\label{line:insert-clean-check}
        else 
            newNode = CreateNode(key, val, curr)
            durCurr = MarkDurable(curr)
            des = $\langle$p.next, p.old$\rangle$
            exp = $\langle$durCurr, null$\rangle$
            iflagCurr = MarkIflag(curr) //\label{line:insert-iflag}
            new = $\langle$newNode, iflagCurr$\rangle$
            if DWCAS(des, exp, new) ://\label{line:insert-dwcas}
                Persist(p, newNode, iflagCurr) //\label{line:insert-persist}
                return true
\end{lstlisting}
 
\myparagraph{Insert} 
The \texttt{Insert} function has two arguments $key$ and $val$ representing the key to insert and its corresponding value.
Any execution of \texttt{Insert} begins by invoking \texttt{Find} yielding the tuple $\langle gp, p, curr\rangle$.
If the key contained $curr$ is the same as the key to be inserted then we return \texttt{false}.
If the key was not found in the list the insert checks that the $p$ is not marked or dflagged using the function $IsClean$.
This check ensures that the insert can fail early if the DWCAS on \cref{line:insert-dwcas} is guaranteed to fail.
If $p$ is dirty then the insert will attempt to help the concurrent update then retry inserting the key.
If $p$ is not dirty then a new node $n$ is created.
The $CreateNode$ function sets the next pointer of $n$ to $curr$, explicitly flushes $n$ to persistent memory and sets the persistence bit in $n.next$.
Next, on \cref{line:insert-iflag}, we create $iflagCurr$ which is a pointer to $curr$ with the iflag bit set.
The insert then attempts a DWCAS to update $p.next$ to $n$ and $p.old$ to $iflagCurr$.
If the DWCAS is successful the insert explicitly flushes $p.next$, performs a persistence fence and sets the persistence bit via the \texttt{Persist} function on \cref{line:insert-persist}.
If the DWACS fails then the insert will retry.

\myparagraph{Remove}
The \texttt{Remove} function has a single argument $key$ representing the key to be removed.
The initial steps of \texttt{Remove} are similar to \texttt{Insert}.
A remove begins by invoking \texttt{Find} which again returns the tuple $\langle gp, p, curr\rangle$ and ensures the same guarantees explained previously.
If $curr.key$ is not $key$ then the remove returns \texttt{false}.
If the key was found then the remove must confirm that both $p$ and $curr$ are clean.
If either is found to be dirty the remove will attempt to help the concurrent update then retry.
If $p$ and $curr$ are clean then $dflagCurr$ is created.
$dflagCurr$ is a dflagged copy of $curr$.
The remove then attempts a DWCAS on \cref{line:remove-dwcas} to update $p.next$ to $dflagCurr$ and $p.old$ to \texttt{null}. 
(When we say that the DWCAS updates $p.old$ to \texttt{null} it is more accurate to say that the DWCAS confirms that $p.old$ is \texttt{null} since the expected value for $p.old$ is also \texttt{null}).
If the DWCAS succeeds then the remove invokes the \texttt{HelpRemove} function to complete the remove then returns \texttt{true} after returning from \texttt{HelpRemove}. 
If the DWCAS fails then the remove retries.
% The insert operation begins by traversing the list by calling the \texttt{Find} function.
% The \texttt{Find} function returns the nodes $gp$, $p$ and $curr$ where $curr.key \geq key$ and there was a time where $p.next$ pointed to $curr$ and $gp.next$ pointed to $p$.

\begin{lstlisting}[name=updates,frame=single, float=t!, label={fig:physical-del-dwcas-updates2}, caption={Remove function for Physical-Delete (DWCAS)}]              
def Remove(key) :
    while true :
        $\langle$gp,p,curr$\rangle$ = Find(key) //\label{line:remove-search}
        cNext = curr.next
        pNext = p.next
        if curr.key $\neq$ key : //\label{line:remove-curr-key-check}
            return false
        if not IsClean(cNext) : //\label{line:remove-clean-check}
            HelpUpdate(p, curr)
        else if not IsClean(pNext) :
            HelpUpdate(gp, p)
        else :
            dflagCurr = MarkDflag(curr) //\label{line:remove-dflag-mask}
            des = $\langle$p.next, p.old$\rangle$
            exp = $\langle$curr, null$\rangle$
            new = $\langle$dflagCurr, null$\rangle$
            if DWCAS(des, exp, new) : //\label{line:remove-dwcas}
                HelpRemove(p, dflagCurr) //\label{line:remove-invoke-help-remove}
                return true
\end{lstlisting}

\myparagraph{Helper Functions}
Update operations rely on a helping mechanism that we implement with the functions \texttt{HelpUpdate}, \texttt{HelpRemove} and \texttt{HelpMarked}.
The \texttt{HelpUpdate} function has two arguments $parent$ and $dirtyNode$ which are both nodes. 
If $dirtyNode$ is dflagged then \texttt{HelpRemove} is invoked.
If $dirtyNode$ is marked then \texttt{HelpMarked} is invoked.
If $dirtyNode$ is clean \texttt{HelpUpdate} returns.

The \texttt{HelpRemove} function is responsible for marking the successor node of a dflagged node.
It takes two arguments $parent$ and $nodeToDel$.
The \texttt{HelpRemove} function begins by ensuring that the next pointer in $nodeToDel$ is persistent.
If $nodeToDel.next$ does not have the persistence bit set then the \texttt{Persist} function is invoked to write $nodeToDel.next$ to persistent memory and set the persistence bit.
Next, we create $markedSucc$ which is a marked copy of $nodeToDel.next$.
On \cref{line:helpRemove-mark-dwcas} we perform ValDWCAS (a DWCAS that returns the last value at the destination) to update $nodeToDel.next$ to $markedSucc$ and confirm that $nodeToDel.old$ is \texttt{null}.
If the value returned by the ValDWCAS matches the expected value $exp$ or if it is marked then we invoke the \texttt{HelpMarked} function then return otherwise we will need to retry.
Before retrying we check if the value returned by the ValDWCAS is dflagged then we invoke the \texttt{HelpRemove} function with the arguments $nodeTodel$ and the unmarked $nodeToDel.next$.

The \texttt{HelpMarked} function is responsible for physically deleting a node.
The \texttt{HelpMarked} function takes two arguments $parent$ and $nodeToDel$
This requires performing a single DWCAS on \cref{line:helpMarked-dwcas} which attempts to atomically update the next pointer of the $parent$ to the unmarked next pointer of the $nodeToDel$ and the old pointer of $parent$ to $nodeTodel$.
If the DWCAS succeeds the \texttt{Persist} function is invoked to write $parent.next$ to persistent memory, set the persistence bit in $parent.next$ and revert $parent.old$ to \texttt{null}.
There is no need to retry this DWCAS since the only way that it can fail is if another process physically deleted $nodeToDel$.

\begin{lstlisting}[frame=single, float=t!, label={fig:physical-del-dwcas-helpers1}, caption={HelpUpdate and HelpMarked Function for Physical-Delete (DWCAS)}]
def HelpUpdate(parent, dirtyNode) :
    $dirtyNext = dirtyNode.next$
    $dirtySucc = \func{UnmarkPtr}(dirtyNext)$
    if $\func{IsDflagged}(dirtyNext)$ : //\label{line:helpUpdate-dflag-check}
        $\func{HelpRemove}(dirtyNode, dirtySucc)$ //\label{line:helpUpdate-invoke-help-remove}
    else if $IsMarked(dirtyNext)$ :
        $\func{HelpMarked}(parent, dirtyNode)$ //\label{line:helpUpdate-helpMarked}

def HelpMarked(parent, nodeToDel) :
    $succ = \func{UnmarkPtr}(nodeToDel.next)$
    $expNext = \func{Markdflag}(nodeToDel)$
    $expNext = \func{MarkDurable}(expNext)$
    $des = \langle parent.next, parent.old \rangle$
    $exp = \langle expNext, null\rangle$
    $new = \langle succ, expNext\rangle$
    if $\func{DWCAS}(des, exp, new)$ : //\label{line:helpMarked-dwcas}
        $\func{Persist}(parent, succ, expNext)$ //\label{line:helpMarked-persist}
\end{lstlisting}

\begin{lstlisting}[frame=single, float=t!, label={fig:physical-del-dwcas-helpers2}, caption={HelpRemove Function for Physical-Delete (DWCAS)}]
def HelpRemove(parent, nodeToDel) :
    while $parent.next$ == $nodeToDel$
        $succ = nodeToDel.next$
        $des = \langle nodeToDel.next, nodeToDel.old \rangle$
        if not $\func{IsDurable}(succ)$ then 
            $\func{Persist}(nodeToDel, succ, null)$
        $durSucc = \func{UnmarkPtr(succ)}$ //\label{line:helpRemove-unmark-succ}
        $durSucc = \func{MarkDurable(durSucc)}$
        $markedSucc = \func{MarkDel}(durSucc)$ //\label{line:helpRemove-mark-mask}
        $exp = \langle durSucc, null\rangle$
        $new = \langle markedSucc, null\rangle$
        $last = ValDWCAS(des, exp, new)$ //\label{line:helpRemove-mark-dwcas}
        $lastMarked = \func{IsMarkedForDel}(last.next)$
        if $last == exp$ or $lastMarked$ : 
            $HelpMarked(parent, nodeToDel)$ //\label{line:helpRemove-helpMarked1}
            return 
        else if $\func{IsDflagged}(last.next)$ :
            $next = \func{UnmarkPtr}(last.next)$
            $\func{HelpRemove}(nodeToDel, next)$
\end{lstlisting}

\myparagraph{Contains}
We offer four different variants of the \texttt{Contains} function.
% The pseudo code for these variants are show in \cref{fig:physical-del-dwcas-search} and \cref{fig:physical-del-dwcas-pa-contains}.
% The specific functions are \texttt{ContainsPersistFree}, \texttt{ContainsPersistLast}, \texttt{ContainsPersistAll} and \texttt{ContainsAsynchPersistAll}.
Each of these functions has a single argument $key$.

The simplest of these variants is the \texttt{ContainsPersistAll}.
In this version we traverse the list until we reach a node containing a key greater than or equal to $key$.
During the traversal, if we find a pointer with a persistence bit set to 0 the \texttt{Persist} function is invoked to write the pointer to persistent memory and set the persistence bit.
Finally, we return \texttt{true} if the terminal node of the traversal contains $key$ and \texttt{false} otherwise.
Unsurprisingly this variant performs poorly in practice.

The \texttt{ContainsAsynchPersistAll} function exploits asynchronous flush instructions to require only a single persistence fence.
In this case, during the traversal any pointer that has its persistence bit set to 0 is asynchronously flushed to persistent memory.
A statically allocated array is used to keep track of the nodes that were asynchronously flushed.
% For simplicity the pseudo code expresses this with \texttt{append} on \cref{line:async-contains-flush-list}. 
(The size of this array depends on the maximum size of the list).
% After the traversal a single persistence fence is performed on \cref{line:async-contains-fence}.
After the traversal a single persistence fence is performed.
Next we iterate the collection of asynchronously flushed nodes.
% If the next pointer of the node has not changed since the flush then we set the persistence bit and revert the old field to \texttt{null} on \cref{line:async-contains-dur-bit}.
If the next pointer of the node has not changed since the flush then we set the persistence bit and revert the old field to \texttt{null}.
Finally, we return \texttt{true} if the terminal node of the traversal contains $key$ and \texttt{false} otherwise.

The \texttt{ContainsPersistLast} function is similar to the search function in \cite{david2018log}.
In this version, we do not flush anything during the traversal.
% Instead, we check that link that led to the the terminal node of the traversal has its persistence bit set on \cref{line:persist-last-contains-dur-bit-check}.
Instead, we check that link that led to the the terminal node of the traversal has its persistence bit set.
If this is not the case then we invoke the \texttt{Persist} function to write the pointer to persistent memory and set the persistence bit.
Again we return \texttt{true} if the terminal node of the traversal contains $key$ and \texttt{false} otherwise.

Finally we have the \texttt{ContainsPersistFree} function. 
The traversal performed by this function is the same as in the \texttt{ContainsPersistLast} function.
If link that led to the the terminal node, $curr$, found during  the traversal has its persistence bit set to 1 then we simply return \texttt{true} if $curr$ contains $key$ and \texttt{false} otherwise on \cref{line:physical-del-pf-contains-ret2}.
If the link does not have its persistence bit set we must read the value stored in the old field of predecessor of $curr$, namely $p.old$.
Unfortunately, there is no way to perform an atomic read of $p.next$ and $p.old$ without using locking or performing a DWCAS.
Both of these options would be expensive.
Fortunately we can avoid this by exploiting the fact that $p.old$ and $p.next$ are atomically updated together.
We must ensure that the value we read for $p.old$ corresponds to the value that we last read for $p.next$. Recall that in this scenario, the \texttt{contains} last read that $p.next$ points to $curr$ and does not have its persistence bit set.
Since we can only atomically read one of $p.old$ and $p.next$ we must reread these fields to ensure that neither field has changed.
This requires reading both $p.old$ and $p.next$ twice then comparing the values.
Lines \ref{line:physical-del-pf-contains-reads-start} through \ref{line:physical-del-pf-contains-reads-end} show these reads and the associated comparisons.
If the first read does not match the second read or if $p.old$ is \texttt{null} then there was a time during the execution of the \texttt{ContainsPersistFree} that the value we read for $p.next$ was persistent so we return \texttt{true} if $curr$ contains $key$ and \texttt{false} otherwise on \cref{line:physical-del-pf-contains-ret2}.
If the values for both reads of $p.old$ and $p.next$ match and $p.old$ is not \texttt{null} then we know that the value we read for $p.old$ corresponds to the value we read for $p.next$ and that $p.next$ is not persistent.
In this case, we check if $p.old$ was iflagged. 
If $p.old$ was iflagged then $p.old$ was set by a concurrent \texttt{Insert} so we can return \texttt{false}.
If $p.old$ was not iflagged we know that $p.next$ is not persistent because of a concurrent remove operation and we know that $p.old$ points to the last persistent value stored in $p.next$.
This also means that there was a time during the execution of the \texttt{ContainsPersistFree} where $curr$ was pointed to by a persistent link so if $curr$ contains $key$ then we can return \texttt{true}.
If $curr$ does not contain $key$ then we return \texttt{true} if the node pointed to by $p.old$ old contains $key$ and \texttt{false} otherwise.

\subsection{Persistent List Correcntess}
\label{appendix:correctenss}
In this section we will provide a proof sketch arguing that the DWCAS implementation of our \physdelList~list is durable linearizable. 
Proving that our \physdelList~list is linearizable and lock-free is very similar to the correctness proof of \cite{fomitchev2004lock}.
It is fairly straightforward to prove lock-freedom and linearizability of the \physdelList~list since the volatile synchronization utilized in the \physdelList~list can be thought of as a simplified version of the linked list in \cite{fomitchev2004lock} which is linearizable and lock-free.
The proof of the DWCAS Harris style list is a natural extension of this proof and the original Harris list correctness proof. 

\myparagraph{Linearization Points}
We will begin by describing how we choose the linearization points for every operation in the \physdelList~list. 
These linearization points are chosen such that updates are linearized at some point after the CPE of the operation.

\begin{definition}[Durable]
At configuration $c$, a pointer $\rho$ is considered \textit{durable} if the durability bit in $\rho$ is set to 1 at $c$.

A node $n$, is considered \textit{durable} if the pointer into $n$ is durable at $c$. 

We say that $n$ is \textit{durably linked} in the volatile data structure if $n$ is durable and $n$ is reachable via a traversal starting from the root at $c$.
\end{definition}

\begin{definition}[Volatile Data Structure]
The \textit{volatile data structure at a configuration $c$}, is the set of nodes that are reachable by a traversal starting from the head.
\end{definition}

\begin{definition}[Volatile Abstract Set]
The \textit{volatile abstract set at a configuration $c$} is the set of keys in nodes in the volatile data structure, at configuration $c$. 
\end{definition}

\begin{definition}[Persistent Data Structure]
Consider an execution $E$ of the \physdelList~list.
Let $c$ be the configuration after the last event in $E$.
Let $E'$ be the crash-recovery extension of $E$.
Let $c'$ be the configuration after the last event in $E'$.
The \textit{persistent data structure at $c$} is equivalent to the volatile data structure at $c'$.
\end{definition}

\begin{definition}[Persistent Abstract Set]
Consider an execution $E$ of the \physdelList~list. 
Let $c$ be the configuration after the last event in $E$.
Let $E'$ be the crash-recovery extension of $E$.
Let $c'$ be the configuration after the last event in $E'$.
The \textit{persistent abstract set at $c$} is the equivalent to the volatile abstract set at $c'$.
\end{definition}

\myparagraph{Insert}
Consider an \texttt{Insert} $\pi$ invoked by process $p$ where the key provided as an argument to $\pi$ is $k$.
Case 1: $\pi$ returns \texttt{false}.
If $\pi$ returns \texttt{false} then there must be a configuration $c$ that exists during the execution of $\pi$ where $k$ is in the persistent abstract set at $c$ and we linearize $\pi$ at a time corresponding to when $k$ is in the persistent abstract set at $c$.
To prove that such a time exists, assume that $k$ is never in the persistent abstract set at any configuration that exists during the execution of $\pi$.
If this were true either $k$ is not in the volatile data abstract set at any configuration that exists between the invocation and response of $\pi$ or $k$ is contained in a node that is in the volatile data structure but not durable at any configuration that exits exists between the invocation and response of $\pi$.
Since the \texttt{insert} returned false $curr.key$ was equal to $k$.
This means that $k$ must have been in the volatile abstract set since we found it via a traversal starting from the head.
The \texttt{Find} invoked by $\pi$ which returned $curr$ guarantees that $curr$ is durable.
This means that there is a configuration that exists between the invocation and response of $\pi$ where $k$ is in the persistent abstract set.
Case 2: $\pi$ returns \texttt{true}.
If $\pi$ returns \texttt{true} then the DWCAS on \cref{line:insert-dwcas} succeeded inserting the node $newNode$ by updating $p.next$.
In this case we linearize $\pi$ at the first successful DWCAS that sets the persistence bit in $p.next$ while $p.next$ points to $newNode$.
This DWCAS must exist at some point before the response of $\pi$ since the DWCAS performed by the \textit{Persist} invoked on \cref{line:insert-persist} can fail if and only if some process other than $p$ successfully performed an identical DWCAS setting the persistence bit in $p.next$.

\myparagraph{Remove}
Consider a \texttt{Remove} $\pi$ invoked by process $p$ where the key provided as an argument to $\pi$ is $k$.
Case 1: $\pi$ returns \texttt{false}.
If $\pi$ returns \texttt{false} then there must be a configuration $c$ that exists during the execution of $\pi$ where $k$ is not in the persistent abstract set at $c$ and we linearize $\pi$ at the time corresponding to when $k$ is not in the persistent abstract set at $c$.
To prove that such a time exists, assume that $k$ is always in the persistent abstract set at any configuration that exists during the execution of $\pi$.
If this were true the node $n$ containing $k$ must always be in the persistent data structure at every configuration that exists between the invocation and response of $\pi$.
If $n$ is in the persistent data structure then $n$ is also in the volatile data structure.
This means that the traversal performed by the \texttt{Find} invoked by $\pi$ will end at $n$ meaning it will be returned as $curr$.
This means that the check on \cref{line:remove-curr-key-check} will fail which is impossible if $\pi$ returns \texttt{false}. 
Case 2: $\pi$ returns \texttt{true}.
If $\pi$ returns \texttt{true} then the DWCAS on \cref{line:helpMarked-dwcas} of the HelpMarked function succeeded in physically deleting the node $nodeToDel$ by updating $parent.next$ to point to the successor of $nodeToDel$.
In this case we linearize $\pi$ at the first successful DWCAS that sets the persistence bit in $parent.next$ while $parent.next$ points to the successor of $nodeToDel$.
This DWCAS must exist at some point before the response of $\pi$ since the DWCAS performed by the \textit{Persist} invoked on \cref{line:helpMarked-persist}  can fail if and only if some process other than $p$ successfully performed an identical DWCAS setting the persistence bit in $p.next$.

\myparagraph{Contains that Persist}
\\Consider a \texttt{ContainsPersistLast} $\pi$ invoked by process $p$ where the key provided as an argument to $\pi$ is $k$.
Case 1: $\pi$ returns \texttt{true}.
If $\pi$ returns \texttt{true} then there must be a configuration $c$ that exists between the invocation and response of $\pi$ where $k$ is in the persistent abstract set at $c$ and we linearize at the time corresponding to when $k$ is in the persistent abstract set at $c$.
The proof is the same as an \texttt{Insert} that returns \texttt{false}.
Case 2: $\pi$ returns \texttt{false}.
If $\pi$ returns \texttt{false} then there must be a configuration $c$ that exists between the invocation and response of $\pi$ where $k$ is not in the persistent abstract set at $c$ and we linearize at the time corresponding to when $k$ is not in the persistent abstract set at $c$.
The proof is the same as an \texttt{Remove} that returns \texttt{false}.

The linearization points and proofs for the other versions of the \texttt{Contains} function that perform persistence events follow the same structure. 

\myparagraph{Contains PersistFree}
\\Consider a \texttt{ContainsPersistFree} $\pi$ invoked by process $p$ where the key provided as an argument to $\pi$ is $k$.
Case 1: $\pi$ returns \texttt{true} at \cref{line:physical-del-pf-contains-ret1}. 
In this case the node $curr$ was found via a traversal starting from the head, the pointer into $curr$ was durable and $curr$ does contain $k$.
This means that there is a time when $k$ was in the persistent abstract set and we linearize at that time.

Case 2: $\pi$ returns \texttt{false} at \cref{line:physical-del-pf-contains-ret1}. 
In this case the node $curr$ was found via a traversal starting from the head and the pointer into $curr$ was durable and $curr$ does not contain $k$.
This means that there is a time when $k$ was not in the persistent abstract set and we linearize at that time.

Case 3: $\pi$ returns \texttt{true} at \cref{line:physical-del-pf-contains-ret2}. 
In this case the node $curr$ was found via a traversal starting from the head and it contains $k$, however, when we first traversed the pointer into $curr$ it was not durable.
After rereading, we found that either the pointer into $curr$ or the old field in the predecessor of $curr$ has changed. 
Apply the same structure as in the proof of \cref{theorem:persist} and note that before the DWCAS the \texttt{Persist} function performs a psync. 
This means that if either of these fields changed then some process must have written $curr$ to persistent memory.
This means that there is a time when $k$ was in the persistent abstract set and we linearize at that time.

Case 4: $\pi$ returns \texttt{false} at \cref{line:physical-del-pf-contains-ret2}. 
This is the same as case 3 except that $curr$ does not contain $k$. This means that there is a time during the execution of $\pi$ when $k$ is not in the persistent abstract set.

Case 5: $\pi$ returns \texttt{false} at \cref{line:physical-del-pf-contains-ret3}.
In this case we have not established the existence of a time where $curr$ is durable and the old field in the predecessor of $curr$ is iflagged.
This means that $curr$ is being inserted by a concurrent operation.
We linearize $\pi$ at any point after its invocation and before the CPE of the concurrent insert.

Case 6: $\pi$ returns \texttt{true} at \cref{line:physical-del-pf-contains-ret4}.
In this case we have not established the existence of a time where $curr$ is durable but the old field in the predecessor of $curr$ is not iflagged and $curr$ contains $k$.
This means that $curr$ is not durable because a concurrent remove has physically deleted the last predecessor of $curr$.
Applying the same argument as in case 3, we establish a time during the execution of $\pi$ where $curr$ is in the persistent data structure meaning $k$ is in the persistent abstract set.

Case 8: $\pi$ returns \texttt{true} at \cref{line:physical-del-pf-contains-ret5}.
In this case we have not established the existence of a time where $curr$ is durable but the old field in the predecessor of $curr$ is not iflagged.
This means that $curr$ is not durable because of an incomplete remove operation. 
Let $p$ be the predecessor of $curr$. 
By \cref{theorem:dur-bit} the node $old$ pointed to by $p.old$ was the last durable value in $p.next$.
If $\pi$ returned \texttt{true} then $p.old$ contained $k$.
Since $p.old$ was not \textit{null} and $p.old$ contains $k$ then there is a time during the execution of $\pi$ when $k$ is in the persistent abstract set and we linearize $\pi$ at that time.
% We linearize $\pi$ at the time when we found $p.old$ corresponded to $p.next$ where $p.next$ pointed to $curr$ since in persistent memory $p.next$ points to $old$.

Case 9: $\pi$ returns \texttt{false} at \cref{line:physical-del-pf-contains-ret5}.
This is the same as case 8 except $old$ does not contain $k$.
% We linearize $\pi$ at the time when we found $p.old$ corresponded to $p.next$ where $p.next$ pointed to $curr$.
This means that there is a time during the execution of $\pi$ when $k$ is not in the persistent abstract set and we linearize $\pi$ at that time.

We prove that the \physdelList~list maintains several invariants.
First, we define some necessary terminology used in the proofs.

% \begin{definition}[Successor]
% At configuration $c$, the successor of the node $n$ is the node pointed to by $n.next$.
% \end{definition}

% \begin{definition}[Predecessor]
% At configuration $c$, the predecessor of the node $n$ is the node $u$ where $u.next$ points to $n$ and $u$ is reachable via a traversal starting from the root.
% \end{definition}

\begin{definition}[Dflagged]
A node $n$ is considered \textit{dflagged} at configuration $c$ if the \textit{dflag} bit in $n.next$ is set to 1 and $n$ is reachable via a traversal starting from the head in $c$.
\end{definition}

\begin{definition}[Logically Deleted]
A node $n$ is considered \textit{logically deleted} at configuration $c$ if the marked bit in $n.next$ is set to 1 and $n$ is reachable via a traversal starting from the root in $c$.
\end{definition}

\begin{definition}[Physically Deleted]
A node $n$ is considered \textit{physically deleted} at configuration $c$ if $n$ cannot be reached by a traversal starting from the root.
\end{definition}

% \begin{invariant}{}
% In a configuration $c$, the key stored in the node pointed to by $n.next$ is strictly greater than the key stored in $n$.
% \end{invariant}

\begin{definition}[Consistency]
The \physdelList~list is consistent with some persistent abstract set $P$ if $\forall k \in P$, $k$ is in a durable node in the volatile data structure.
\end{definition}

\begin{invariant}{}
A node is never both logically deleted and dflagged.
\label[invariant]{invar:reserved-bits}
\end{invariant}

\begin{invariant}{}
Once a node is marked its next pointer never changes. 
\label[invariant]{invar:marked}
\end{invariant}

\begin{invariant}{}
If the node $n$ is in the volatile data structure at configuration $c$ and $n$ is logically deleted in $c$, then the predecessor of $n$ is dflagged and the successor of $n$ is not logically deleted in $c$.
\label[invariant]{invar:2-step-remove}
\end{invariant}

\begin{invariant}
In a configuration $c$, if $n.next$ is not durable then the $n.old$ is the last durable pointer stored in $n.next$ 
% not \texttt{null}.
\label[invariant]{invar:dur-bit}
\end{invariant}

% \begin{invariant}{}
% At configuration $c$, the next pointer in a node is durable iff the successor of $n$ is in the persistent data structure at $c$.
% \label[invariant]{invar:dur-ptr}
% \end{invariant}

% \begin{invariant}{}
% The volatile abstract set at configuration $c$ is a subset of the persistent abstract set at $c$.
% \label[invariant]{invar:persistent-set-1}
% \end{invariant}

\begin{invariant}{}
Consider an execution $E$ of the \physdelList~list.
Let $c$ be the configuration after the last event in $E$.
Let $P_c$ denote the persistent abstract set at $c$ and $V_c$ denote the volatile abstract set at $c$.
$P_t \setminus V_t$ is a subset of the keys that were part of remove operations that have no response in $E$ and $V_t \setminus P_t$ is a subset of the keys that were part of insert operations that have no response in $E$. 
\label[invariant]{invar:persistent-set-2}
\end{invariant}

\begin{theorem}
\Cref{invar:reserved-bits} always holds for any configuration produced by an execution of the \physdelList~list.
\label{theorem:reserved-bits}
\end{theorem}
\begin{proof}
The invariant is trivially true for an empty list.
When a new node is created the next pointer in the node is not marked or dflagged.
This cannot change until the node is linked in the list such that it is reachable via a traversal from the head.
% The next pointer of any node is atomically updated by the DWCAS on \cref{line:async-contains-dur-bit} of the \textit{ContainsAsynchPersistAll} function, \cref{line:physical-del-persist-dwcas} of the \texttt{Persist} function, \cref{line:insert-dwcas} of the \texttt{Insert} function, \cref{line:remove-dwcas} of the \textit{Remove} function, \cref{line:helpMarked-dwcas} of the \textit{HelpMarked} function, \cref{line:helpRemove-mark-dwcas} of the \textit{HelpRemove} function and \cref{line:physical-del-persist-dwcas} of the \textit{Persist} function.
The next pointer of any node is atomically updated by the DWCAS on \cref{line:physical-del-persist-dwcas} of the \texttt{Persist} function, \cref{line:insert-dwcas} of the \texttt{Insert} function, \cref{line:remove-dwcas} of the \textit{Remove} function, \cref{line:helpMarked-dwcas} of the \textit{HelpMarked} function, \cref{line:helpRemove-mark-dwcas} of the \textit{HelpRemove} function and \cref{line:physical-del-persist-dwcas} of the \textit{Persist} function.
None of these DWCAS will update a node to be both marked and dflagged.
\end{proof}

\begin{theorem}
\cref{invar:marked} always holds for any configuration produced by an execution of the \physdelList~list.
\end{theorem}
\begin{proof}
No DWCAS will ever modify a marked pointer.
The expected value of every DWCAS in the form $\langle n, o\rangle$ is always explicitly defined such that $n$ is unmarked.
\end{proof}

\begin{theorem}
\label{theorem:dur-bit}
\Cref{invar:dur-bit} always holds for any configuration produced by an execution of the \physdelList~list. 
\end{theorem}
\begin{proof}
The old field in any node $n$ is updated to a non-\texttt{null} value by the DWCAS on \cref{line:insert-dwcas} of the \texttt{Insert} function or the DWCAS on \cref{line:helpMarked-dwcas} of the \texttt{HelpMarked} function.
Let $\langle e_1, \texttt{null}\rangle$ be the expected value of the DWCAS that updated $n.old$ to a non-\texttt{null} value and $\langle x_1, x_2\rangle$ be the new value. 
In both cases the $x_2$ is explicitly defined to be equal to $e_1$ and $e_1$ is explicitly defined to be a durable pointer.
\end{proof}

\begin{theorem}
\Cref{invar:2-step-remove} always holds for any configuration produced by an execution of the \physdelList~list. 
\end{theorem}
\begin{proof}
If a node $n$ is in the volatile data structure at configuration $c$ then $n$ is reachable via a traversal starting from the head.
This means that the both the successor and predecessor of $n$ are in the volatile data structure at $c$.
$n$ can only become logically deleted by the DWCAS on \cref{line:helpRemove-mark-dwcas} of the \texttt{HelpRemove} function.
If an execution of \texttt{HelpRemove} reached this DWCAS then the condition of the \texttt{while} loop in \texttt{HelpRemove} must have been true.
This means that there was a time during the execution of \texttt{HelpRemove} where the predecessor of $n$ was dflagged.
The only DWCAS that removes the dflag is the DWCAS on \cref{line:helpMarked-dwcas} of the \texttt{HelpMarked} function.
This DWCAS also physically deletes the successor of the dflagged node. 
This means that if the DWCAS succeeds, $n$ will no longer be in the volatile data structure and while $n$ is in the volatile data structure its predecessor is dflagged.
In order for the successor of $n$ to be logically deleted the DWCAS on \cref{line:helpRemove-mark-dwcas} of \texttt{HelpRemove} must succeed. 
The expected value of this DWCAS is explicitly unflagged and unmarked on \cref{line:helpRemove-unmark-succ}.
Since $n$ is logically deleted the DWCAS will never succeed.
\end{proof}

\begin{theorem}
\label{theorem:persist}
The DWCAS on \cref{line:physical-del-persist-dwcas} of the \texttt{Persist} function invoked by process $p$ will fail if and only if some process other than $p$ successfully completes an identical DWCAS.
\label{theorem:persist-failed}
\end{theorem}
\begin{proof}
Consider an execution of the \texttt{Persist} function by process $p$.
Let $n$ be the node such that the destination field of the DWCAS on \cref{line:physical-del-persist-dwcas} of the \texttt{Persist} by process $p$ is $\langle n.next, n.old\rangle$.
Call this DWCAS $d$.
The expected value of $d$ is $\langle x, o\rangle$ where $x$ is a non-durable pointer.
$d$ sets the persistence bit in $n.next$ and reverts $n.old$ to \texttt{null}.
% Since $p$ is executing \texttt{Persist} there is configuration that exists between
% a time when $n.next$ was not durable.
Assume $p$ fails $d$.
This means that the actual value of $n.next$ is either $x'$ where $x'$ is equivalent to $x$ with the persistence bit set or some other value $y \neq x$.

In this case some other process $p'$ must have concurrently updated $n$ when $n.next$ was $x$.
This means that $p'$ must have performed an update involving $n$ or an asynchronous-persist-all \texttt{Contains}.

If $p'$ successfully completed the the DWCAS in asynchronous persist-all \texttt{Contains} while $n.next$ was $x$ then it is easy to see that this DWCAS is identical to $d$.
The DWCAS in every other function (excluding \texttt{Persist}) has an expected value that is explicitly defined to be a durable next pointer and a \texttt{null} old pointer.
This means that these DWCAS will always fail if the $n.next$ is not durable. 
Since $x$ is not durable, any process that attempts a DWCAS on $n$ in any of the update or helper functions will fail. 
In each case, if the process fails the DWCAS it must retry the operation.
This will eventually lead to the process invoking \texttt{Persist} with the first two arguments being $n$ and $x$.
Thus one of these processes will successfully complete a DWCAS identical to $d$.
\end{proof}

% \begin{theorem}
% \cref{invar:persistent-set-1} holds for any configuration produced by an execution of the \physdelList~list.
% \end{theorem}
% \begin{proof}
% 
% \end{proof}

% \begin{theorem}
% Consider an execution $E$ of the \physdelList~list.
% A single \texttt{Insert} operation $\pi$ that has an invocation in $E$ will add at most one key to the persistent abstract set and removes no keys from the persistent abstract set.
% \end{theorem}
% \begin{proof}
% If $\pi$ completes in $E$ and the response value was \texttt{false} then $\pi$ will not have a CPE.
% Any psync performed by $\pi$ will be helping another update operation.
% This means that the persistent abstract set is not directly $\pi$.
% \end{proof}

\begin{theorem}
Update operations of the \physdelList~list change the volatile abstract set exactly once.
\label{theorem:vol-set-once}
\end{theorem}
\begin{proof}
We can identify each of the atomic instructions that update the volatile abstract set.
Note that unsuccessful update operations do not change either the volatile data structure nor the volatile abstract set.
For insert operations the volatile abstract set is updated at the success of the DWCAS on \cref{line:insert-dwcas} of the \texttt{Insert} function.
For remove operations the volatile abstract set is updated at the success of the DWCAS on \cref{line:helpMarked-dwcas} of the \texttt{HelpMarked} function.
These atomic instructions also correspond to the points at which the volatile data structure is updated.
\end{proof}

\begin{theorem}
Update operations of the \physdelList~list change the volatile abstract set before changing the persistent abstract set.
\label{theorem:volitile-before-persist}
\end{theorem}
\begin{proof}
Consider an execution $E$ of the \physdelList~list.
Let $c$ be the configuration after the last event in $E$.
The persistent data structure at $c$ contains all of the durable nodes in volatile data structure at $c$ as well as any non-durable node in volatile data structure at $c$ where the node is not durable because of a remove operation.

Update operations modify the volatile abstract set before they update the persistent abstract set.
Insert operations modify the volatile abstract set at the success of the DWCAS on \cref{line:insert-dwcas} of the \texttt{Insert} function.
Remove operations modify the volatile abstract set at the success of the DWCAS on \cref{line:helpMarked-dwcas} of the \texttt{HelpMarked} function.
In both cases the DWCAS by definition always happens before the CPE of the update operation.

If a crash occurs after a remove operation $\pi$ succeeds the DWCAS on \cref{line:helpMarked-dwcas} of \texttt{HelpMarked} but before its CPE then the persistent abstract set will still contain the key that $\pi$ removed from the volatile abstract set.
More precisely, let $R$ be the set of remove operations that were incomplete in $E$.
Let $k_{\pi}$ be the key that was input to $\pi \in R$.
For any $\pi \in R$ that completes the DWCAS on \cref{line:helpMarked-dwcas} of \texttt{HelpMarked} but the CPE of $\pi$ has not occurred $k_{\pi}$ will not be in the volatile abstract set at $c$ but $k_{\pi}$ will be in the persistent abstract set at $c$. 

Similarly, if a crash occurs after an insert operation $\pi$ succeeds the DWCAS on \cref{line:insert-dwcas} of \texttt{Insert} but before its CPE then the the volatile abstract set will contain the key inserted by $\pi$ but the persistent abstract set will not.
\end{proof}

\begin{theorem}
\Cref{invar:persistent-set-2} holds for any configuration produced by an execution of the \physdelList~list.
\end{theorem}
\begin{proof}
This follows from \cref{theorem:volitile-before-persist}.
\end{proof}

\begin{theorem}
\texttt{Insert}, \texttt{Remove} and \texttt{Contains} are lock-free. 
\label{theorem:lock-free}
\end{theorem}
\begin{proof}
\texttt{Insert}, \texttt{Remove} and \texttt{Contains} all perform the same traversal of the list beginning from the head.
In each case the traversal is a \texttt{while} loop that after finding a node contain a key greater or equal to the search key.
Each iteration of the loop executes O(1) steps.
If no concurrent process successfully performs the DWCAS on \cref{line:insert-dwcas} of \texttt{Insert} or \cref{line:helpMarked-dwcas} of \texttt{HelpMarked} then the volatile data structure and the volatile abstract set is unchanged.
This means that the loop will terminate in a finite number of steps.
If a concurrent process successfully performs the DWCAS on \cref{line:insert-dwcas} of \texttt{Insert} then the loop could require one more iteration.
This means that the traversal will complete in a finite number of steps or another process completes a DWCAS updating the volatile abstract set.

In the case of \texttt{Contains} operations or the \texttt{Find} function, the instructions performed after the traversal execute O(1) steps.

Update operations have a retry loop.
One iteration of the loop executes O(1) steps. 
For \texttt{Insert} another iteration of the loop is required if the node $p$ node returned by the \texttt{Find} function is marked or flagged or if the DWCAS on \cref{line:insert-dwcas} fails.
If the DWCAS in \texttt{Insert} fails then some other concurrent process must have completed a DWCAS on $p$.
If $p$ is dirty then the \texttt{Insert} will help the concurrent update.
The helper function \texttt{HelpMarked} executes O(1) steps.
The helper function \texttt{HelpRemove} has a retry loop. 
A single iteration of this loop executes O(1) steps.
Another iteration of the loop is required if the DWCAS on \cref{line:helpRemove-mark-dwcas} fails and the value returned is not marked.
Since the condition of the loop checks that the node $nodeToDel$ is still the successor of the node $parent$, failing the DWCAS means that some other concurrent operation successfully performed a DWCAS on $parent$.
If no other process performs a DWCAS on $parent$ then the loop in \texttt{HelpRemove} will terminate.

For \texttt{Remove} another iteration of the loop is required if the node $p$ or the node $curr$ returned by the \texttt{Find} function is marked or flagged or if the DWCAS on \cref{line:remove-dwcas} fails.
If the DWCAS in \texttt{Remove} fails then some other concurrent process must have completed a DWCAS on $p$.
If $p$ is dirty or $node$ is dirty the \texttt{Remove} will help complete the concurrent update requiring the same steps as described for the case of \texttt{Insert}.
\end{proof}

\begin{theorem}
The \physdelList~list is always consistent with some persistent abstract set $P$.
\label{theorem:consistency}
\end{theorem}
\begin{proof}
From \cref{theorem:vol-set-once} we know that the volatile abstract set is changed exactly once by any update operation. 
It follows from \cref{theorem:volitile-before-persist} that the \physdelList~list will be consistent with $P$ since every completed successful update operation will be reflected in $P$ along with some of the pending update operation.
\end{proof}

\begin{theorem}
% The \physdelList~list is durable linearizable and lock-free.
The PD~list is durable linearizable and lock-free.
\label{theorem:list-correctness}
\end{theorem}

\Cref{theorem:consistency} and \cref{theorem:lock-free} and the way in which we choose the linearization points for the \physdelList~list collectively prove \cref{theorem:list-correctness}.

Proving the SWCAS implementations is significantly more involved due to the extra volatile memory synchronization. 
A proof of the SWCAS implementations could be constructed similarly to the proof of the \physdelList~list.
The main difference would be the proof of \cref{theorem:dur-bit} since a node's old pointer is updated independently of its next pointer.

\fi
\section{Evaluation}
\label{section:eval}

We present an experimental analysis of our persistent list compared to existing persistent lists on various workloads.
We test our variants of the \texttt{contains} operation separately meaning no run includes more than one of the variants.\footnote{Source code: \url{https://gitlab.com/Coccimiglio/setbench}}
To distinguish between our implementations of the \texttt{contains} operation we prefix the names of our persistent list algorithms with the abbreviation of a \texttt{contains} variant (for example PFLD refers to one of our persistent lists which utilized only Persistence-Free searches and the Logical-Deletion algorithm).
\ifx\fullpaper\undefined
Due to space constraints we only present the best performing implementations of our persistent list.
\else
We present the best performing implementations of our persistent list.
\fi
We test the performance of these lists in terms of throughput (operations per second). 
We also examine the psync behaviour of these algorithms. 
Specifically, we track the number of psyncs that are performed by searches and the number of psyncs that are performed by update operations. 

All of the experiments were run on a machine with 48 cores across 2 Intel Xeon Gold 5220R 2.20GHz processors which provides 96 available threads (2 threads per core and 24 cores per socket). 
The system has a 36608K L3 cache, 1024K L2 cache, 64K L1 cache and 1.5TB of NVRAM. 
The NVRAM modules installed on the system are Intel Optane DCPMMs.
We utilize the same benchmark as \cite{brown2020non} for conducting the empirical tests.
Keys are accessed according to a uniform distribution.
We prefill the lists to 50\% capacity before collecting measurements.
Each test consisted of ten iterations where each individual test ran for ten seconds.
The graphs show the average of all iterations.
Libvmmalloc was the persistent memory allocator used for all algorithms.

\myparagraph{Throughput}
Figure \ref{fig:vs-others-clflush} shows the throughput of our best persistent list variants compared to the existing algorithms.
Since the DWCAS implementation of our list out performed the SWCAS implementation we compare only our DWCAS implementations.
\SOFTList~ performs best when there is high contention in read dominant workloads and consistently best for non-read dominant workloads.

\noindent\textbf{\underline{Lesson learned:}} Persisting more information in update operations is generally more costly but persistence free searches do not seem to provide major performance improvements.

\begin{figure}[t!]
    \begin{subfigure}{1.0\linewidth}
        \centering
        \includegraphics[width=0.95\linewidth]{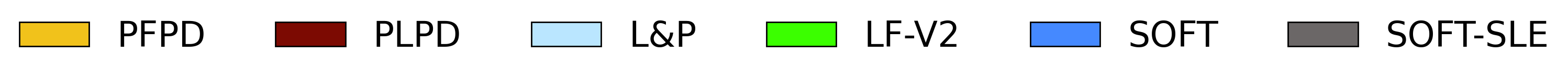}\\
        Legend for all plots in Section \ref{section:eval}
    \end{subfigure}
    \begin{subfigure}{0.24\linewidth}
        \centering
        \includegraphics[width=1\linewidth]{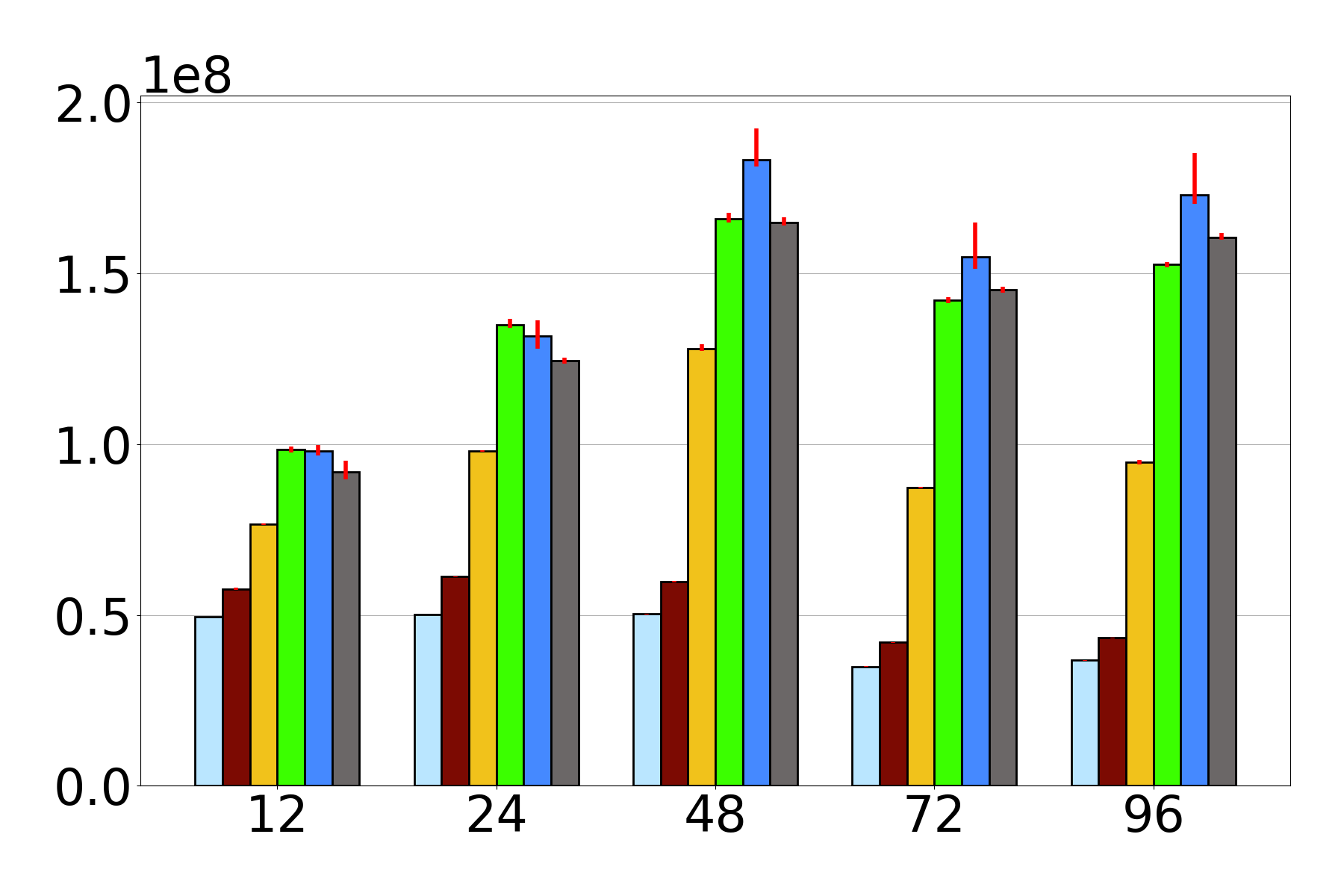}
        \vspace{-7mm}
        \captionsetup{justification=centering}
        \caption{99\% Search K=50}
        \label{subfig:vs-others-a}
    \end{subfigure}
    \begin{subfigure}{0.24\linewidth}
        \centering
        \includegraphics[width=1\linewidth]{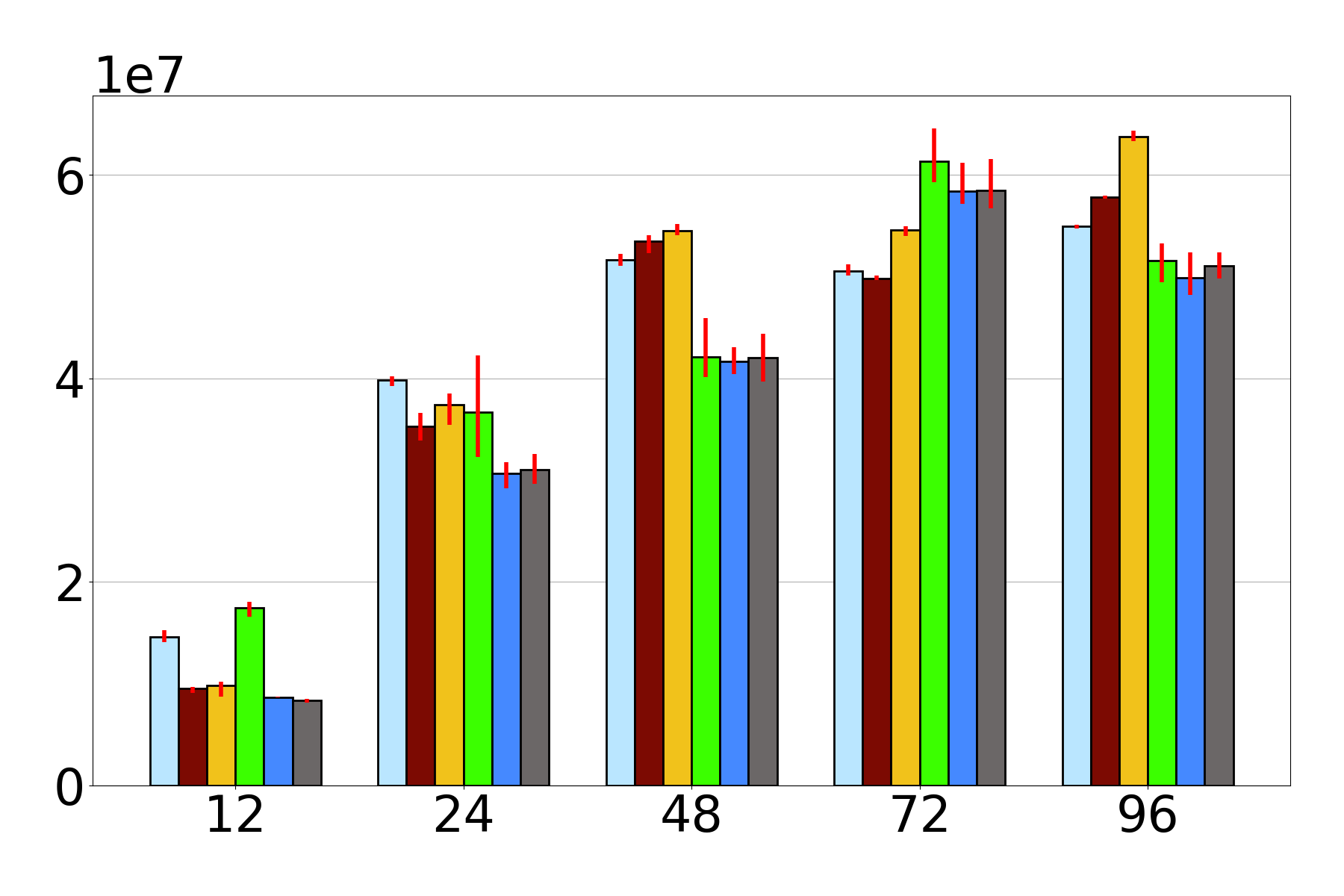}
        \vspace{-7mm}
        \captionsetup{justification=centering}
        \caption{99\% Search K=500}
        \label{subfig:vs-others-c}
    \end{subfigure}
    \begin{subfigure}{0.24\linewidth}
        \centering
        \includegraphics[width=1\linewidth]{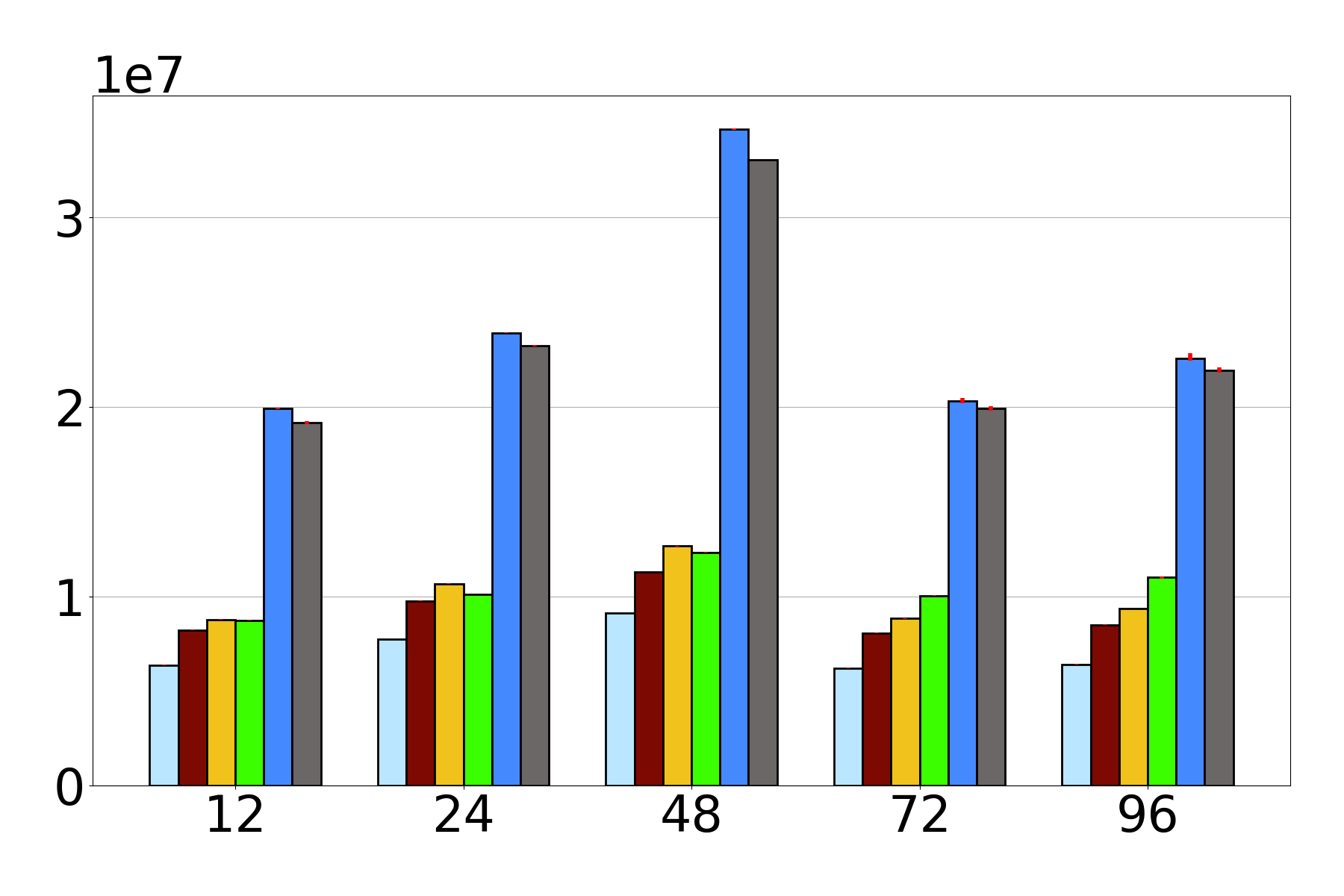}
        \vspace{-7mm}
        \captionsetup{justification=centering}
        \caption{50\% Search K=50}
        \label{subfig:vs-others-d}
    \end{subfigure}
    \begin{subfigure}{0.24\linewidth}
        \centering
        \includegraphics[width=1\linewidth]{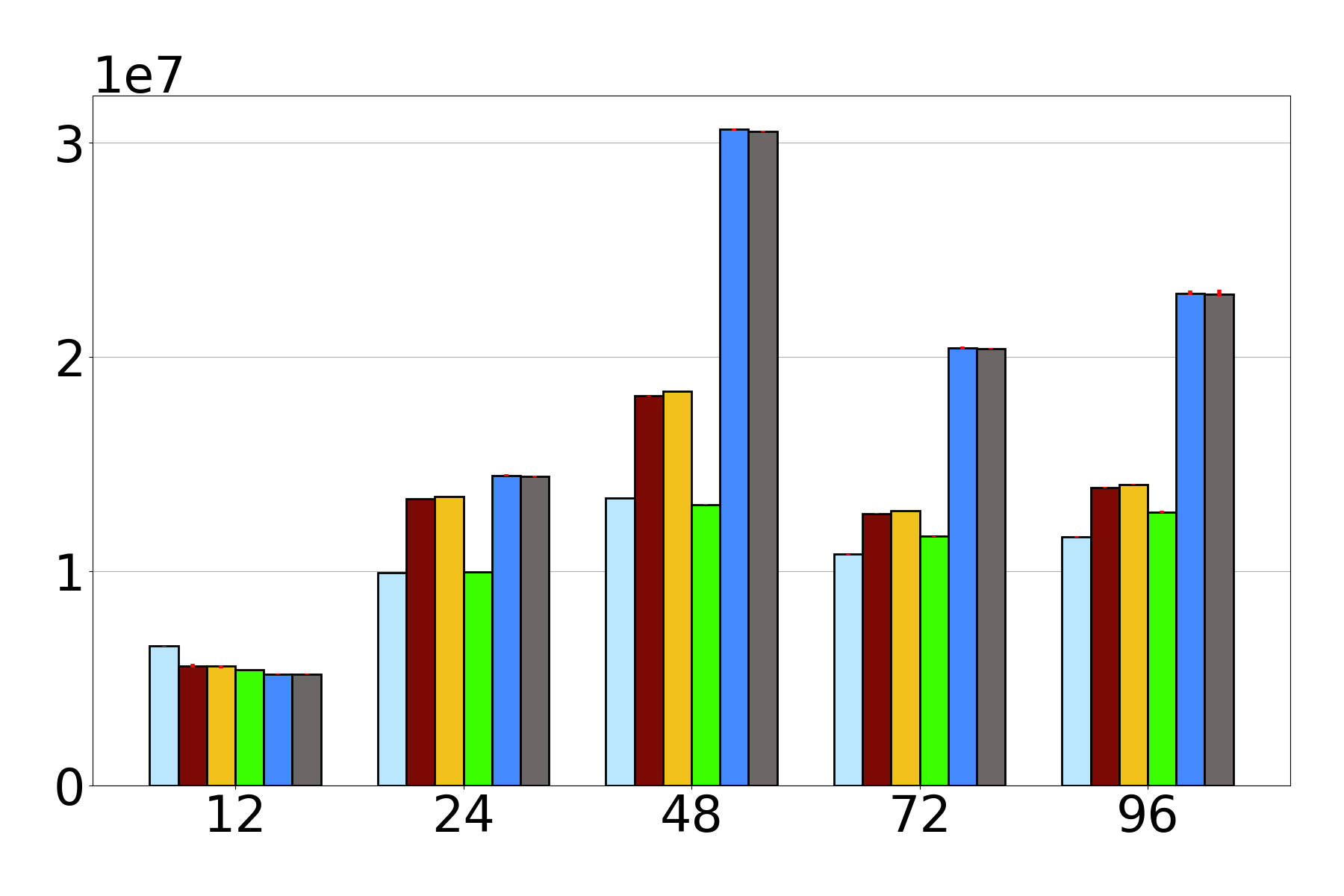}
        \vspace{-7mm}
        \captionsetup{justification=centering}
        \caption{50\% Search, K=500}
        \label{subfig:vs-others-f}
    \end{subfigure}
    \caption{Persistent list throughput. X-axis: number of concurrent threads. Y-axis: operations/second. K is the list size.}
    \label{fig:vs-others-clflush}
\end{figure}

\myparagraph{Psync Behaviour}
The recent trend to persist less data structure state has influenced implementations of persistent objects focused on minimizing the amount of psyncs required per operation.
We know that SLE linearizable algorithms cannot have persistence-free searches.
From \cite{cohen2018inherent} we also know that update operations require at least 1 psync.
Of the persistent lists that we consider, the persistent lists from \cite{zuriel2019efficient} are unique in that the the maximum number of psyncs per update operation is bounded.
To better understand the cost incurred by psyncs, we track the number of psyncs performed by read-only operations (searches) and the number of psyncs performed by update operations. 
Note that for updates this includes unsuccessful updates which might not need to perform a psync.
\Cref{fig:flushes-per-op} shows the average number of psyncs per search and the average number of psyncs per update operation.
We observe that searches rarely perform a psync in any of the algorithms that do not have persistence-free searches.
On average, update operations do not perform more than the minimum number of required psyncs.

\noindent\textbf{\underline{Lesson learned:}} Algorithmic techniques such as \textit{persistence bits} for reducing the number of psyncs are highly effective. On average, there are very few redundant psyncs in practice.

\myparagraph{Recovery}
It is not practical to force real system crashes in order to test the recovery procedure of any algorithm.
It is possible that one could simulate a system crash by 
running the recovery procedure as a standalone algorithm on an artificially created memory configuration.
This is problematic because the recovery procedure of a durable linearizable algorithm is often tightly coupled to some specific memory allocator (this is true of the existing algorithms that we consider).  
This makes a fair experimental analysis of the recovery procedure difficult.
It is easier to describe the worst case scenario for recovering the data structure for each of the algorithms.
To be specific, we describe the worst case persistent memory layout produced by the algorithm noting how this relates to the performance of the recovery procedure.

The \LFList~list does not persist data structure links.
As a result, there is no way to efficiently discover all valid nodes meaning the recovery procedure might need require traversing all of the memory.
The allocator utilized by Zuriel et al partitions memory into \textit{chunks}.
We can construct a worse case memory layout for the recovery procedure as follows:
Suppose that we completely fill persistent memory by inserting keys into the list.
Now remove nodes such that each chunk of memory contains only one node at an unknown offset from the start of the chunk.
To discover all of the valid nodes the recovery procedure must traverse the entire memory space.
The \SOFTList~list also does not persist data structure links.
The requirements of the recovery procedure for \SOFTList~list is the same as the \LFList~list.
We can construct the worst case memory layout for the recovery procedure in the same way as we did for the \LFList~list yielding the same asymptotic time complexity.
The \LPList~list can utilize an empty recovery procedure.
The actual recovery procedure for the list implemented by the authors of \cite{david2018log} does extra work related to memory reclamation. 

We utilize DWCAS and asynchronous flush instructions to achieve a minimum of one psync per \texttt{insert} operation.
There are some subtleties with this implementation that result in a recovery complexity which is O($N+n$) for a list containing $N$ nodes and a maximum of $n$ concurrent processes. 
Implementations that use SWCAS (or DWCAS allowing for a minimum of two psyncs per \texttt{insert}) can utilize an empty recovery procedure. 

\noindent\underline{\textbf{Lesson learned:}} If structure is persisted, recovery can be highly efficient. 
Without any persisted structure, recovery must traverse large regions (or even all) of shared memory. 

\begin{figure}[t!]
    \begin{subfigure}{0.24\linewidth}
        \centering
        \includegraphics[width=1\linewidth]{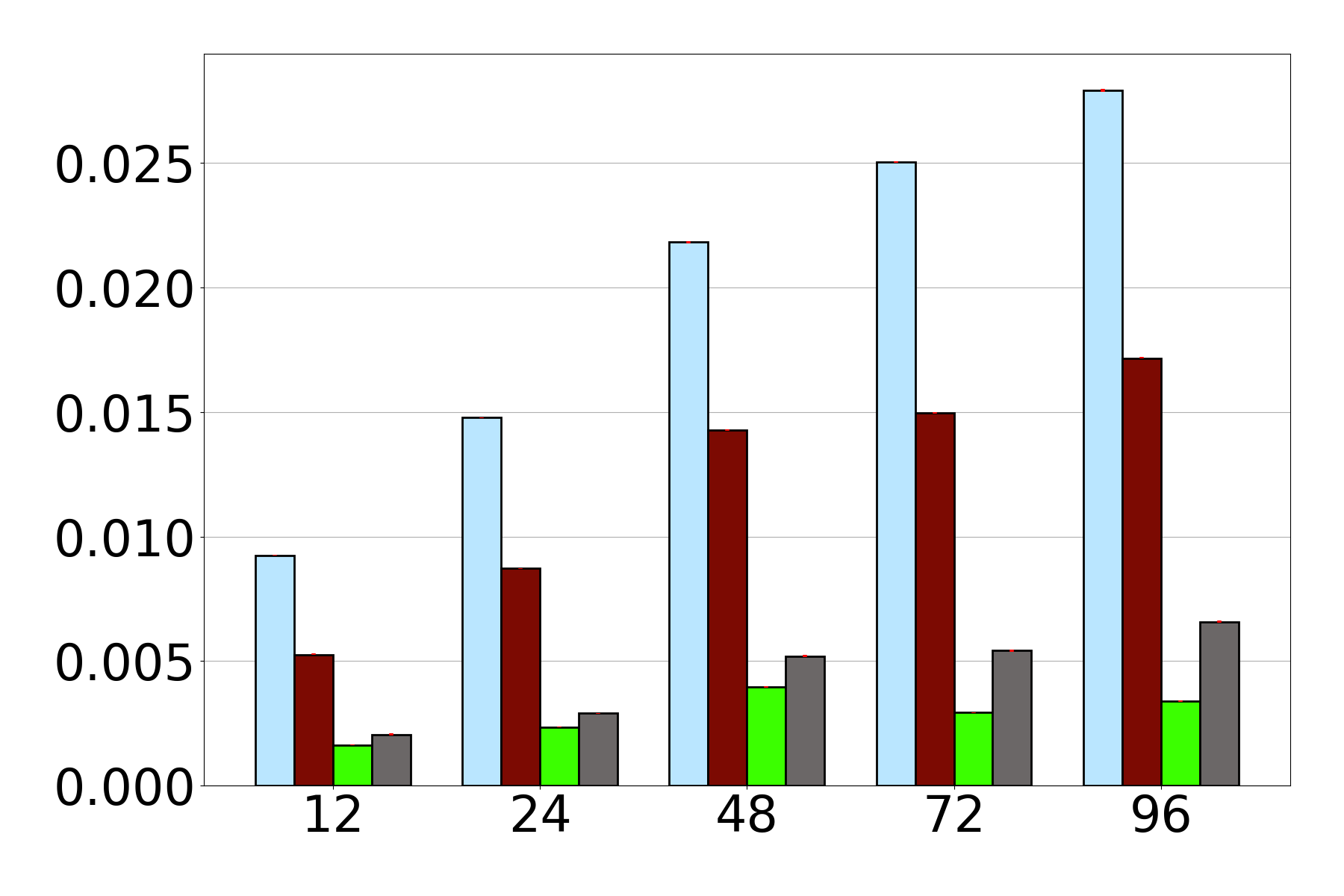}
        \vspace{-7mm}
        \caption{99\% Search}
        \label{subfig:flushes-per-contains-a}
    \end{subfigure}
    \begin{subfigure}{0.24\linewidth}
        \centering
        \includegraphics[width=1\linewidth]{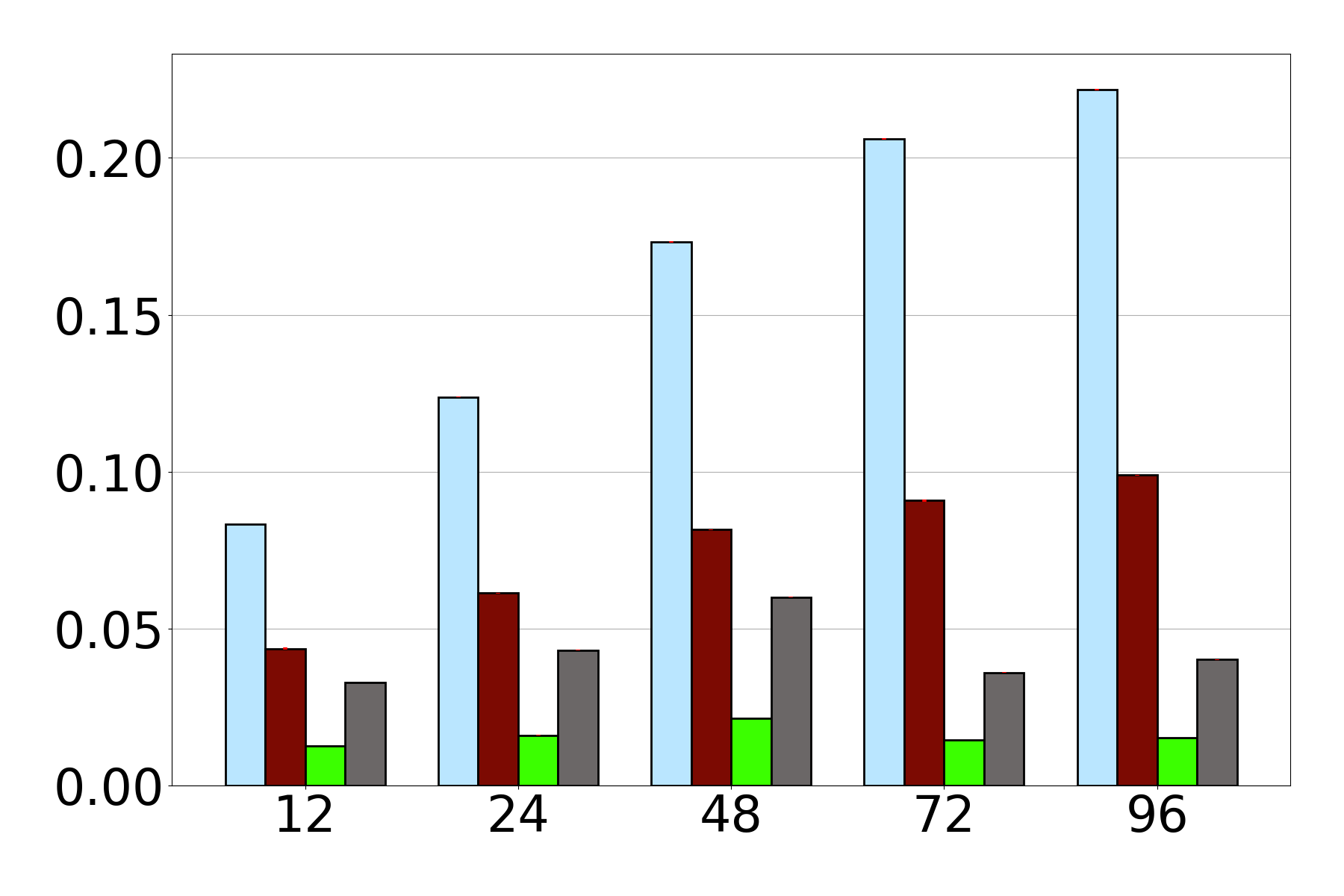}
        \vspace{-7mm}
        \caption{50\% Search}
        \label{subfig:flushes-per-contains-b}
    \end{subfigure}
    \begin{subfigure}{0.24\linewidth}
        \centering
        \includegraphics[width=1\linewidth]{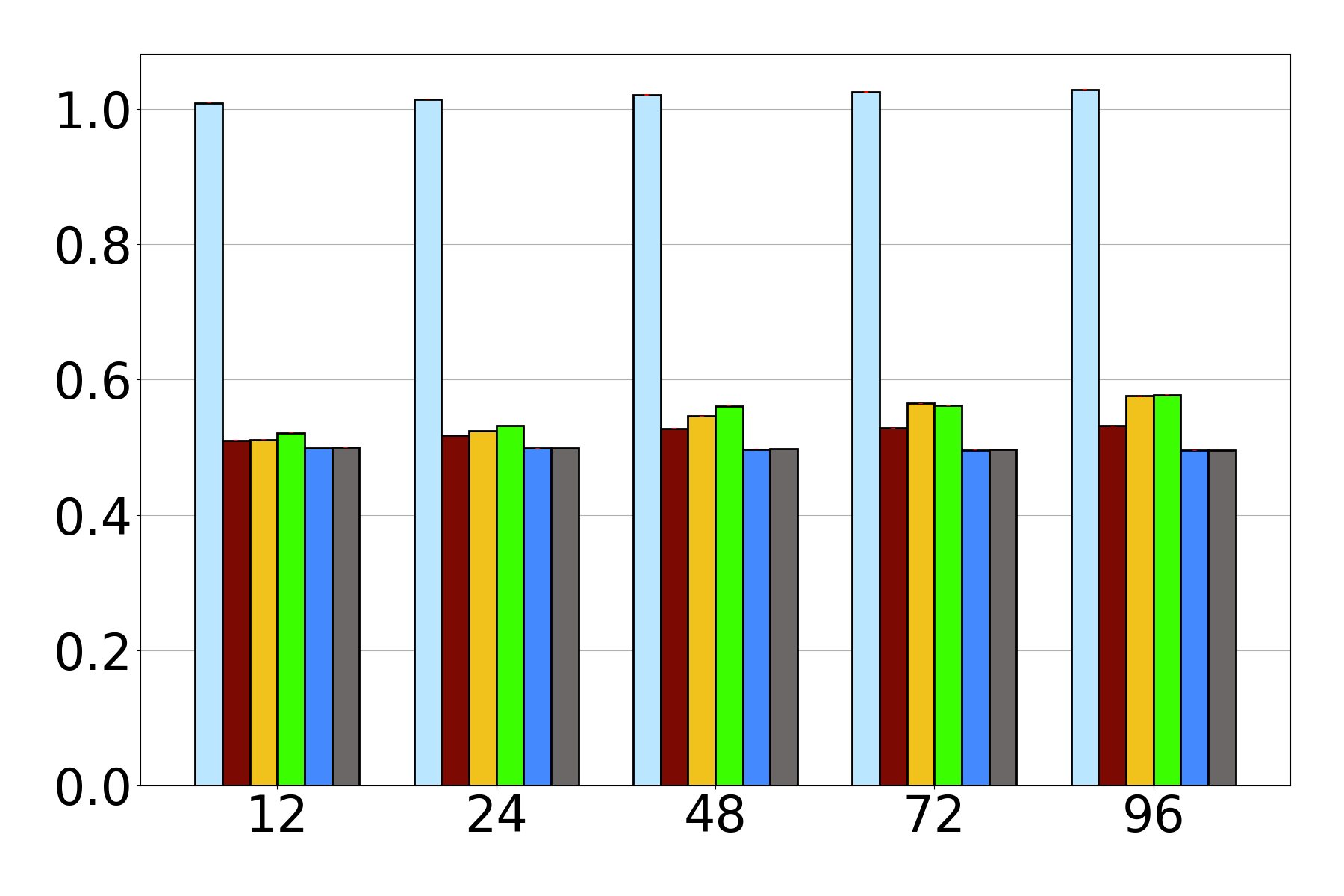}
        \vspace{-7mm}
        \caption{99\% Search}
        \label{subfig:flushes-per-update-a}
    \end{subfigure}
    \begin{subfigure}{0.24\linewidth}
        \centering
        \includegraphics[width=1\linewidth]{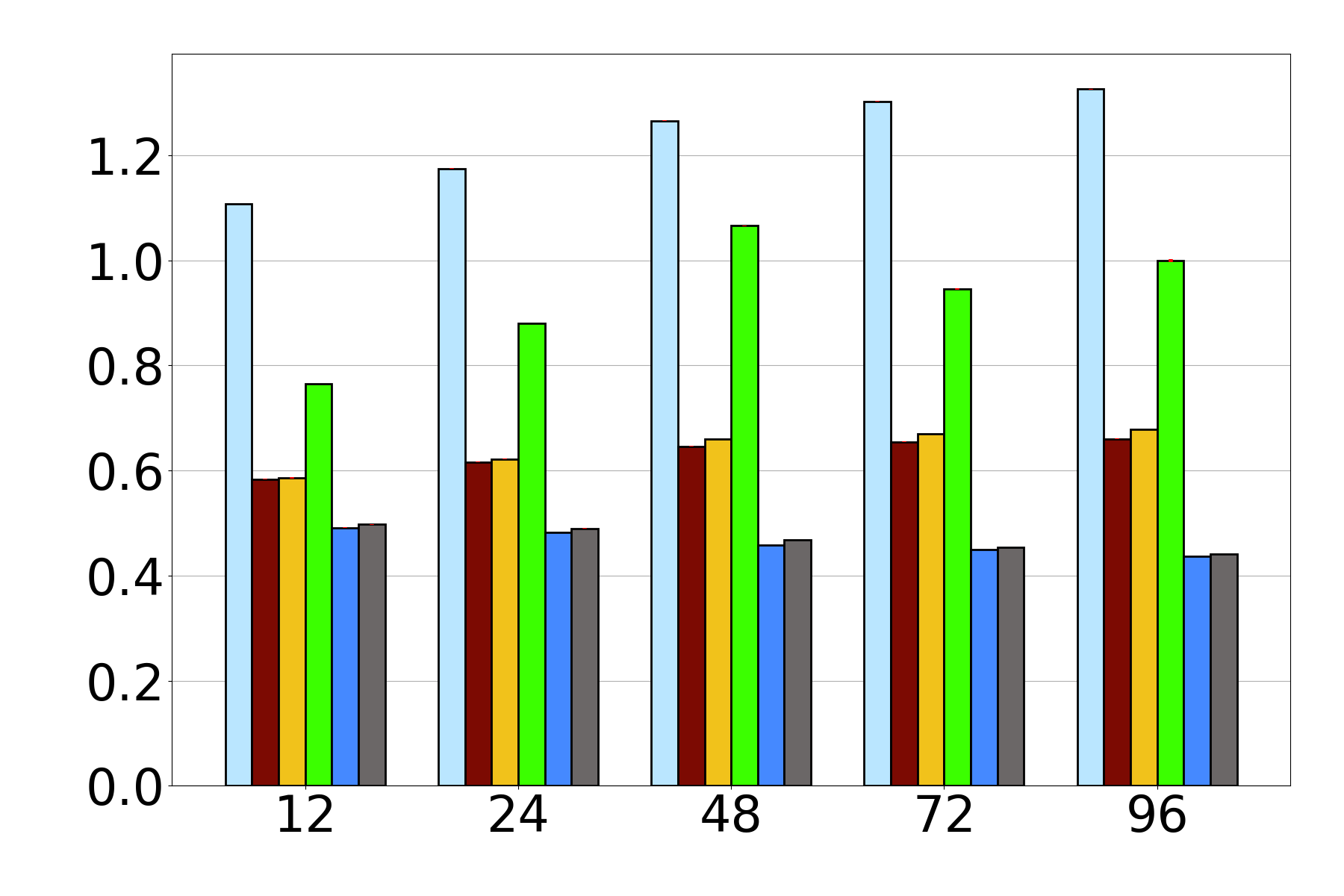}
        \vspace{-7mm}
        \caption{50\% Search}
        \label{subfig:flushes-per-update-b}
    \end{subfigure}
    \caption{Psync Behaviour. X-axis: number of concurrent threads. (a), (b) Y-axis: average psyncs/search, (c), (d) Y-axis: average psyncs/update. List size is 50.}
    \label{fig:flushes-per-op}
    \vspace{-3mm}
\end{figure}

\myparagraph{SLE linearizable vs. Durable linearizable Sets}
We have seen that there exists a theoretical separation between SLE linearizable and durable linearizable objects.
For persistent lists we observe that this separation does not lead to significant performance differences in practice. 
4 of the algorithms (Figure~\ref{fig:vs-others-clflush}) are SLE linearizable.
Specifically, our PLPD list, the L\&P list, LF list, and SOFT-SLE.
The SOFT list and our PFPD list which both use persistence-free searches are durable linearizable.
The high cost of a psync and the impossibility of persistence-free searches in a SLE linearizable lock-free algorithm would suggest that the SLE linearizable algorithms that we test should perform noticeably worse. 
In practice, it is true that for most of the tested workloads, the algorithms that have persistence-free searches perform best (primarily SOFT).
However, for many workloads, performance of SLE linearizable algorithms are comparable to the durable linearizable algorithms.
In fact, for some workloads, the SLE linearizable lists perform better than the durable linearizable alternatives.

\noindent\underline{\textbf{Lesson learned:}} SLE linearizable algorithms can be fast in practice, despite theoretical tradeoffs.
\section{Discussion}
We prove that update operations in durable linearizable lock-free sets will perform at least one redundant psync.
We motivate the importance of ensuring limited effect for sets and defined strict limited effect (SLE) linearizability for sets.
We prove that SLE linearizable lock-free sets cannot have persistence-free reads.
We implement several persistent lists and evaluate them rigorously.
Our experiments demonstrate that SLE linearizable lock-free sets can achieve comparable or better performance to durable linearizable lock-free sets despite the theoretical separation.
For the algorithms and techniques that we examined, supporting persistence-free reads is what separates the durable linearizable sets from the SLE linearizable.
However, the SLE linearizable sets rarely perform a psync during a read.
For those researchers that value ensuring limited effect for sets but are unsure about the performance implications, we recommend beginning with SLE linearizable implementations since a SLE linearizable implementation may not have much overhead and it may be sufficient for the application.
Our work also exposes that psync complexity is not a good predictor of performance in practice, thus motivating need for better metrics to compare persistent objects.

In this work we focused specifically on sets because we wanted to understand the psync complexity of a relatively simple data structure like sets.
We think that there is clear potential to generalize our theoretical results to other object types or classes of object types and perform similar empirical analysis of persistent algorithms for those objects, thus bridging the gap between theory and practice.

\subsubsection{Acknowledgements} 
This work was supported by: the Natural Sciences and Engineering Research Council of Canada (NSERC) Collaborative Research and Development grant: CRDPJ 539431-19, the Canada Foundation for Innovation John R. Evans Leaders Fund with equal support from the Ontario Research Fund CFI Leaders Opportunity Fund: 38512, NSERC Discovery Launch Supplement: DGECR-2019-00048, NSERC Discovery Program grant: RGPIN-2019-04227, and the University of Waterloo.

\bibliographystyle{splncs04}
\bibliography{refs}

\end{document}